\documentclass[journal]{IEEEtran}
\usepackage{pdfsync,srcltx} 
\usepackage{amsmath}
\usepackage{amsfonts}
\usepackage{amssymb}
\usepackage{amsthm}

\usepackage{gensymb}
\usepackage{enumerate}
\usepackage{subcaption}

\usepackage{enumitem}
\newlist{sublist}{enumerate}{1}
\setlist[sublist,1]{label=(\alph*)}

\usepackage{etoolbox}

\usepackage{graphicx}
\usepackage[dvipsnames,usenames]{color}
\usepackage[colorlinks, linkcolor=Blue, filecolor =Red, citecolor=RoyalPurple]{hyperref}
\usepackage{lineno}
\usepackage{commath}
\modulolinenumbers[5]

\newtheorem{definition}{Definition}

\newtheorem{proposition}{Proposition}
\newtheorem{lemma}{Lemma}
\newtheorem{theorem}{Theorem}
\newtheorem{assumption}{Assumption}

\theoremstyle{remark} 
\newtheorem{comment}{Comment}

\def\rd#1{{\color{red}{#1}}}

\def\rte{{\textsc{rt}}}
\def\rtevb{\ensuremath{\eta_{\textsc{rt}}}}

\def\hvac{{\textsc{hvac}}}
\def\b{{(b)}}

\def\eqdef{\ensuremath{:=}}
\def\defeq{\ensuremath{=:}}

\newcommand{\whoami}{NR}

\ifdefstring{\whoami}{PB}{
\graphicspath{{/Users/pbarooah/Dropbox/papers/17_Raman_RTE/figures/}}%
}{
\graphicspath{{figures/}}
}

\def\pb#1{\footnote{\rd{pb:#1}}}
\def\pb#1{}
\def\nsr#1{\footnote{\rd{nsr:#1}}}
\def\nsr#1{}

\def\version{Arxiv}

\newboolean{ArxivVersion}
\ifthenelse{\equal{\version}{Arxiv}}
{\setboolean{ArxivVersion}{true} 
}
{\setboolean{ArxivVersion}{false} 
}

\usepackage{xparse}

\ExplSyntaxOn
\cs_new:Nn \pb_prop_gset_bykeys:Nn
{
	\prop_clear:N \l__pb_temp_prop
	\keys_set:nn { pb/propbykey } { #2 }
	\prop_gset_eq:NN #1 \l__pb_temp_prop
}
\cs_generate_variant:Nn \pb_prop_gset_bykeys:Nn { c }

\keys_define:nn { pb/propbykey }
{
	unknown .code:n = \prop_put:Nxn \l__pb_temp_prop { \l_keys_key_tl } { #1 }
}

\cs_generate_variant:Nn \prop_put:Nnn { Nx }

\NewDocumentCommand{\setupcollaborator}{mm}
{
	\prop_new:c { g_collaborator_#1_prop }
	\pb_prop_gset_bykeys:cn { g_collaborator_#1_prop } { #2 }
}

\NewDocumentCommand{\selectcollaborator}{m}
{
	\prop_map_inline:cn { g_collaborator_#1_prop }
	{
		\tl_set:cn { ##1 } { ##2 }
	}
}
\ExplSyntaxOff

\setupcollaborator{PBmac}
{
	NAME=Prabir Mac , laptop or mac mini,
	DiCEbibPATH=/Users/pbarooah/Dropbox/DiCElab_bib,
	ACbibPATH=/Users/pbarooah/Dropbox/austin_bib,
	NSRbibPATH=/Users/pbarooah/Dropbox/NSRbib,
}
\setupcollaborator{NRwindows}
{
	NAME=NR on his windows computer,
	DiCEbibPATH=../../DiCElab_bib,
	ACbibPATH=../../austin_bib,
	NSRbibPATH=../../NSRbib,
}
\setupcollaborator{NRlinux}
{
	NAME=NR on his linux computer,
	DiCEbibPATH=../../DiCElab_bib,
	ACbibPATH=../../austin_bib,
	NSRbibPATH=../../NSRbib,
}

\selectcollaborator{NRwindows} 

\begin{document}
	
\title{On the round-trip efficiency of an HVAC-based virtual battery}
	
\author{\IEEEauthorblockN{Naren Srivaths Raman and Prabir Barooah,~\IEEEmembership{Member,~IEEE}}
	\thanks{The research reported here has been partially supported by an NSF grant (award no. 1646229) and DOE grant (titled ``virtual batteries'') under the GMLC program.}
\thanks{N. Raman and P. Barooah are with the Mechanical and Aerospace Engineering Department, University of Florida, Gainesville, FL 32611 USA (e-mail: narensraman@ufl.edu).}}
	
\maketitle
	
\begin{abstract}
Flexible loads, especially heating, ventilation, and air-conditioning (HVAC) systems can be used to provide a battery-like service to the power grid by varying their demand up and down over a baseline. Recent work has reported that providing virtual energy storage with HVAC systems lead to a net loss of energy, akin to a low round-trip efficiency (RTE) of a battery. In this work we rigorously analyze the RTE of a virtual battery through a simplified physics-based model. We show that the low RTEs reported in recent experimental and simulation work are an artifact of the experimental/simulation setup. When the HVAC system is repeatedly used as a virtual battery, the asymptotic RTE is 1. Robustness of the result to assumptions made in the analysis is illustrated through a simulation case study.
\end{abstract}
	


\section{Introduction} 
There is a growing recognition that the
power demand of most electric loads is flexible, and this flexibility can be exploited to provide ancillary services to the grid by varying the demand up and down over a baseline~\cite{makarov2008assessing,linbarmeymid:2015}. To the grid they
appear to be providing the same service as a
battery~\cite{chsamwu:2017}. Such a load, or collection of loads, can therefore be called Virtual Energy Storage (VES) systems or virtual batteries.  

Consumers' quality of service (QoS) must be maintained by these virtual batteries. When heating, ventilation, and air-conditioning (HVAC) systems are used for VES, a key QoS measure is indoor
temperature. 
Another important QoS measure is the total energy consumption. Continuously varying the power consumption of
loads around a baseline may lead to a net reduction in the efficiency
of energy use, causing the load to consume more energy in the long
run. If so, that will be analogous to the virtual battery having a
round-trip efficiency (RTE) less than unity. Electrochemical batteries
also have a less-than-unity round-trip efficiency due to various losses~\cite{luowandoocla:2015}. 

The aim of this paper is to analyze the RTE of VES system comprised
of HVAC equipment in commercial buildings. The inspiration for this paper comes
from~\cite{beihisbac:2015} and its follow-on
work~\cite{linmatjohhisbac:2017}. To the best of our knowledge,  the
article~\cite{beihisbac:2015} is the first to provide experimental
data on the  round-trip efficiency of buildings providing virtual
energy storage from an experiment carried out at a building in the Los Alamos National Laboratory (LANL) campus. The average
RTE (over many tests) reported in~\cite{beihisbac:2015}  was less than $0.5$. These values are quite low compared
to that for electrochemical batteries, which vary
from 0.75 to 0.97 depending on the electrochemistry~\cite{luowandoocla:2015}.
If the RTE estimates in~\cite{beihisbac:2015} are representative, that bodes ill for the
use of building HVAC systems to be virtual batteries. 

This paper provides an analysis of the RTE of an HVAC based VES system using a simple physics-based model. We establish that the RTE in fact approaches 1 when the HVAC system is repeatedly used as a virtual battery for many cycles. The low RTE values seen in the LANL experiments was due to the fact that the experiment was run for one cycle.

\subsection{Literature review and statement of contribution}
In the experiments reported in \cite{beihisbac:2015},  fan power was varied in
an approximately square wave fashion with a time period of 30
minutes in a $\sim30,000$ m$^2$ building
in the LANL campus. After one cycle of the square wave, the climate
control system was re-activated to bring the building temperature back to
its baseline value. It was observed that the control system had
to expend a considerable amount of additional energy in the recovery
phase in almost all the tests performed. This loss was expressed as a
round-trip efficiency less than unity. In a small number of
tests, the RTE was observed to be greater than unity. The mean RTE
observed from all the tests was in the order of $0.5$.  

In the experiments reported in the article~\cite{linbarmeymid:2015}, fan power in a $\sim 3700$ m$^2$ building at the University of
Florida (UF) campus was varied to track Pennsylvania-New Jersey-Maryland's (PJM) RegD signal~\cite{PJMmanual12}. When the VES controller was turned off at the end
of the experiment, no large transient was observed in either power or temperature; see Figure 8 of~\cite{linbarmeymid:ACC:2015}. A more
recent paper that also presented results from experiments in a test building at Lawrence Berkeley National Laboratory (LBNL) in which HVAC fan power was varied to track RegD, observed similar behavior~\cite{Vrettos_etal_part2:2017}. In fact~\cite{Vrettos_etal_part2:2017} reported a slight \emph{decrease} in energy use compared to the
baseline. Unlike the LANL experiments, the UF and LBNL experiments involved higher frequency variation in the HVAC power, of time scales shorter than 10 minutes. 

The paper~\cite{linmatjohhisbac:2017} attempted to explain the
experimental observations in~\cite{beihisbac:2015} by conducting simulations. They also examined the effect of
several model parameters and sources of experimental uncertainty such
as imprecise knowledge of baseline power consumption. They were able
to replicate several trends observed in the LANL experiments, but there were also significant differences.  

The purpose of this paper is to rigorously analyze the RTE of a VES
system that is based on commercial-building HVAC equipment and to
determine factors that affect the RTE. In that sense, our goal is
similar to that of \cite{linmatjohhisbac:2017}. In contrast to
\cite{linmatjohhisbac:2017}, which explored the effect of many factors
on the RTE by simulation alone, we focus on deriving results for a
limited set of conditions for which provable results can be provided. Following
\cite{linmatjohhisbac:2017}, we also use a simplified physics-based model
of a building's temperature dynamics and power consumption instead of
using a simulation software so that rigorous analysis is possible.

This paper makes two main contributions to the nascent literature on
the RTE of HVAC-based virtual batteries. The first contribution is to
show that the RTE values much smaller than unity that were reported in
prior work were an artifact of the experimental/simulation set up. In
particular, the HVAC system underwent only ``one demand-response
event'' in~\cite{beihisbac:2015}, i.e., one period of a square-wave
power variation. The simulation study~\cite{linmatjohhisbac:2017} also
focused on that situation and observed similar values of the RTE. It
did explore multiple demand-response events, in which the RTE was
found to be close to 1. These events, however, were chosen in a
particular manner that are unlikely to occur in practice. We focus on a general case in which an HVAC-based virtual battery undergoes $n$ repeated cycles of a square-wave power variation. 
We  show through rigorous analysis that the RTE approaches $1$ as $n \to \infty$. When $n$ is small, especially when $n=1$, we show that the RTE can indeed be larger or smaller than 1 depending on the time period of the reference signal.  

Second, we explicitly define terms and concepts that are standard for electrochemical batteries, such as ``state of charge'', but not yet for HVAC-based virtual batteries. Some of these terms were used---even implicitly defined---in prior work~\cite{beihisbac:2015,linmatjohhisbac:2017}. We believe future studies on RTE of virtual batteries will benefit from the definitions proposed here.

The rest of the paper is organized as follows. Section \ref{sec:defs} describes the terminology and definitions needed for the sequel. Section \ref{sec:model} describes the HVAC system model used. Section \ref{sec:zero-mean-powerref} provides analysis of RTE, and Section \ref{sec:simulations} provides a numerical case study that demonstrates robustness of the analysis to the assumptions. Section \ref{sec:conclusion} summarizes the conclusions.

\section{Definitions and other preliminaries}\label{sec:defs}
\subsection{Round-trip efficiency of an electrochemical battery}
The state of charge (SoC) of a battery, which we denote by $S_B(t)$, is
defined as \cite{ng:2009enhanced}
\begin{linenomath*}
	\begin{align}
	\label{eq:SoC-bat-integral}
	S_B(t) = S_B(0) + \frac{1}{Q_0}\int_{0}^t I_B(t)dt,
	\end{align}
\end{linenomath*}
where $I_B(t)$ is the current drawn by the battery (positive during charging and negative during
discharging) and $Q_0$ is its maximum charge (in Coulomb). The SoC is a number between 0 and 1. It is more convenient to use power drawn (or discharged) instead of current in~\eqref{eq:SoC-bat-integral}. For simplicity, we assume the voltage across the battery is constant, $V_0$, so the power drawn by the battery from the grid is $P_B(t) \eqdef V_0I_B(t)$. Eq.~\eqref{eq:SoC-bat-integral} then becomes $ S_B(t) = S_B(0) + \frac{1}{Q_0V_0}\int_{0}^t P_B(t)dt$. Differentiating, we get
\begin{linenomath*}
\begin{align}
  \label{eq:SoC-bat-diff-form}
  C_0\dot{S}_B(t) = P_B(t)
\end{align}
\end{linenomath*}
where $C_0   = Q_0V_0$.

\begin{definition}[Complete charge-discharge]\label{def:cd}
We say a battery has undergone a complete  charge-discharge during a
time interval $[t_i, \; t_f]$ if $SoC(t_i) =SoC(t_f)$. The time
interval  $[t_i, \; t_f]$ is called a complete charge-discharge interval.
\end{definition}
The qualifier ``complete'' does not mean the SoC reaches 1 or 0. It
only means the SoC comes back to where it started from.

\begin{definition}[RTE]\label{def:rte}
Suppose a battery undergoes a complete
charge-discharge over a time interval $[0, t_{cd}]$. Let $t_c$ be the length of time
during which the battery is charging and $t_d$ be the length of time
during which the battery is discharging so that $t_c+t_d=t_{cd}$. The
round-trip efficiency (RTE) of the battery, denoted by $\eta_\rte$,
during this interval is
\begin{linenomath*}
\begin{align}
\label{eq:rte-batt}
\eta_\rte \triangleq \frac{E_d}{E_c} = \frac{-\int_{t_d}P_{B}(t)dt}{\int_{t_c}P_{B}(t)dt},
\end{align}
\end{linenomath*}
where $E_d$ is the energy released by the battery to the grid during discharging, $E_c$ is the energy consumed by the battery from the grid during charging, and $\int_{t_c}$ (resp., $\int_{t_d}$) denotes integration performed over the charging times (resp., discharging times).  
\end{definition}
Notice that by convention $P_B(t)<0$ means the battery is releasing power to the grid at
time instant $t$. In general, the RTE depend on many factors including
how a particular SoC was achieved~\cite{sprenkle:2017life}. For simplicity, we ignore those effects and use \eqref{eq:rte-batt} to define \emph{the} RTE of the battery.

\subsection{Round-trip efficiency of an HVAC-based VES system}

We now consider an HVAC system whose power demand is artificially varied from its baseline demand to provide virtual energy storage. \emph{The power consumption of the virtual battery, $\tilde{P}$, is defined as the deviation of the electrical power consumption of the HVAC system from the baseline power consumption: 
  \begin{align}\label{eq:def-P-vb}
    \tilde{P}(t) := P_\hvac(t) - P_\hvac^\b(t),
  \end{align}
where $P_\hvac^\b$ is the \emph{baseline} power consumption of the HVAC system, defined as the power the HVAC system needs to consume to maintain a baseline indoor temperature $T^\b$.}

To make a connection between a real battery and a virtual battery, consider the simple dynamic model of a building's temperature:
\begin{linenomath*}
\begin{align}
\label{eq:T-dynamics-basic}
C\dot{T}(t) = q(t),
\end{align}
\end{linenomath*}
where $C$ is the heat capacity of the building (J/K) and $q$ is the
net heat influx rate (J/s), which is the combined effect of heat gain
of the building from solar, outdoor weather, occupants, and the HVAC
system.  Comparing \eqref{eq:T-dynamics-basic} and
\eqref{eq:SoC-bat-diff-form}, we see that indoor temperature, $T(t)$,
and the SoC of an electrochemical battery, $S_B$, are analogous. Just
as the SoC of a real battery must be kept between 0 and 1, the
temperature of a building must be kept between a minimum value,
denoted by  $T_L$ (low), and a maximum value, denoted by $T_H$ (high), to
ensure QoS. We therefore define the SoC of an HVAC-based virtual
battery  as follows.
\begin{definition}[SoC of a VES system]\label{def:soc-vb}
The SoC of an HVAC-based VES system with indoor temperature $T$
is the ratio $\dfrac{T_H-T}{T_H-T_L}$ where $[T_L, \; T_H]$ is the allowable range of indoor temperature.
\end{definition}

The definition of a complete charge-discharge interval of a virtual
battery is the same as that for a battery: Definition~\ref{def:cd}, with SoC as defined in Definition~\ref{def:soc-vb}. The round-trip efficiency of the virtual battery, denoted as $\rtevb$, is also the same as that of a battery (Definition~\ref{def:rte}), with power consumption of the battery, $P_B(t)$, replaced by power consumption of the virtual battery, $\tilde P(t)$. Thus, 
\begin{linenomath*}
\begin{align}
\label{eq:rte-def-vb}
\rtevb= \frac{-\int_{t_d} \tilde P(t)dt}{\int_{t_c} \tilde P(t)dt} =\frac{-\int_{t_d}(P_\hvac(t)-P_\hvac^\b(t))dt}{\int_{t_c}(P_\hvac(t)-P_\hvac^\b(t))dt}
\end{align}
\end{linenomath*}

\ifArxivVersion
\subsection{Charging vs. change in SoC}
A comment on the implication of Definition \ref{def:soc-vb} is in order. Whether an increase in SoC is accompanied by an increase in the VES
system's power consumption depends on the baseline
condition. Imagine the scenario when the HVAC system provides net
cooling.  When the VES system charges, i.e., $\tilde{P}>0$,
additional cooling is provided to the building (see
\eqref{eq:def-P-vb}), and the temperature decreases over baseline,
thereby increasing the SoC according to the definition
above. Similarly, when it discharges, i.e., $\tilde{P}<0$, less
cooling is provided and temperature increases over baseline,
thereby lowering the SoC.  If the HVAC system is in the
heating mode, the opposite occurs. Charging ($\tilde{P}>0$) means
more heating, increase in temperature and therefore lowering of the SoC, and vice versa for discharging. Although that may appear contrary to intuition
based on electrochemical batteries, we believe it is sensible since
charging (resp., discharging) of a battery, real or virtual, should
correspond to positive (resp., negative) power draw from the grid
since the grid operator needs to use the same language in
communicating with all batteries. SoC, on the other hand, is a local
concern that only affects the battery operator, and distinct notions
of SoC for distinct types of batteries are not unreasonable. 
\fi

\section{Model of an HVAC-based VES system}\label{sec:model}
Figure~\ref{fig:AHU-chiller-schematic} shows the idealized
variable-air-volume (VAV) HVAC system
under study. The only devices that consume significant amount of
electricity are the supply air fan and the chiller. The energy consumed by the chilled water pump
motors is assumed to be negligible.

\begin{figure}[ht]
  \centering
  \includegraphics[width=1\linewidth]{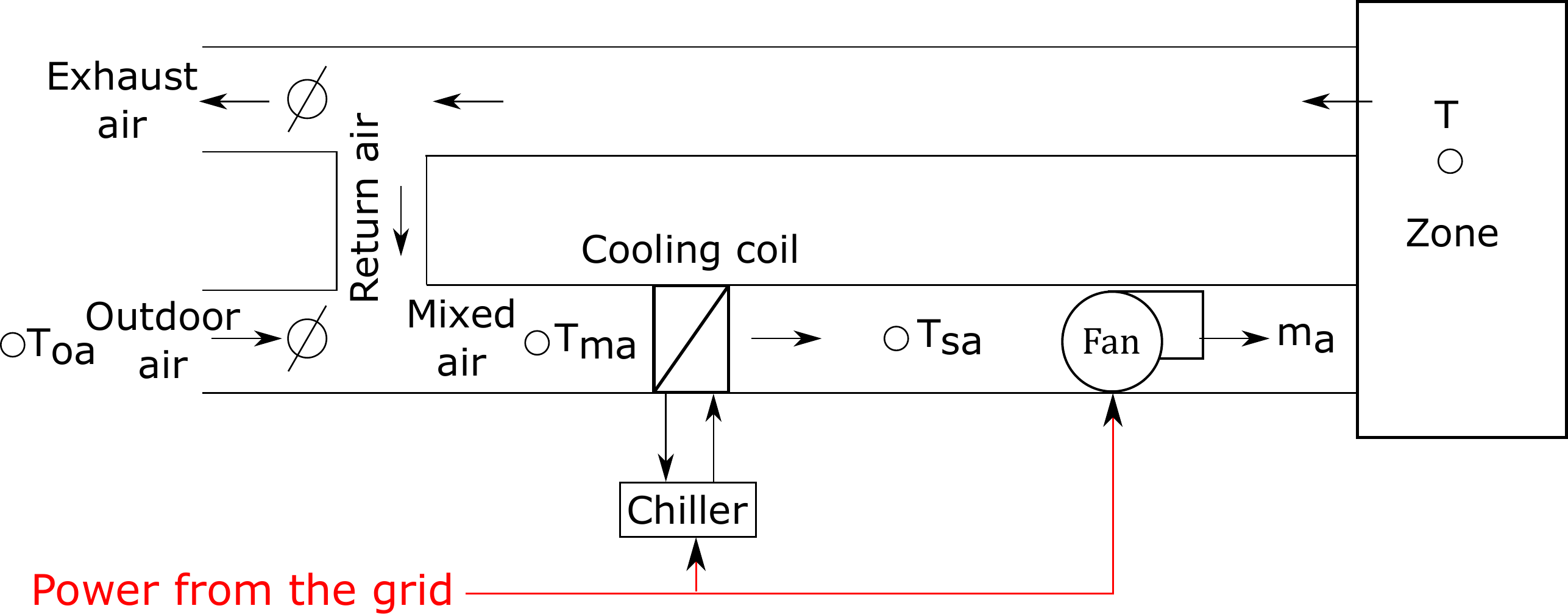}
  \caption{Simplified schematic of a commercial variable-air-volume
    HVAC system.}
  \label{fig:AHU-chiller-schematic}
\end{figure}

In the sequel, $m_a$ denotes the air flow rate\footnote{Customarily air flow rate is denoted by
  $\dot{m}$. Since the notation $\dot{x}$ is used for state
  derivatives (as in $\dot{x} = f(x,u)$), whereas air flow rate is an
  input ($u$) and not a state ($x$), we avoid the ``dot'' notation for
  air flow rate.}. Under baseline conditions, a climate control system
determines the set point for the airflow rate, and the fan speed is
varied to maintain that set point. 

The HVAC system is converted to a VES system with the help of an
additional control system, which we denote by  ``VES controller''. The VES controller modifies the set point of the air flow rate (that is otherwise decided by the climate control system) so that the power consumption of the virtual battery, $\tilde{P}(t)$, tracks an exogenous reference signal, $\tilde{P}^r(t)$. 
We assume that the VES controller is perfect; it can determine the variation in airflow required to track a power deviation reference exactly. \ifArxivVersion Figure \ref{fig:VEScontroller} illustrates the action of the VES controller. When the HVAC system is not providing VES service, the VES controller is turned off: $\tilde{P}(t)\equiv0$. In other words, the building is under baseline operation.  
 
\begin{figure}[ht]
	\centering
	\includegraphics[width=0.55\linewidth]{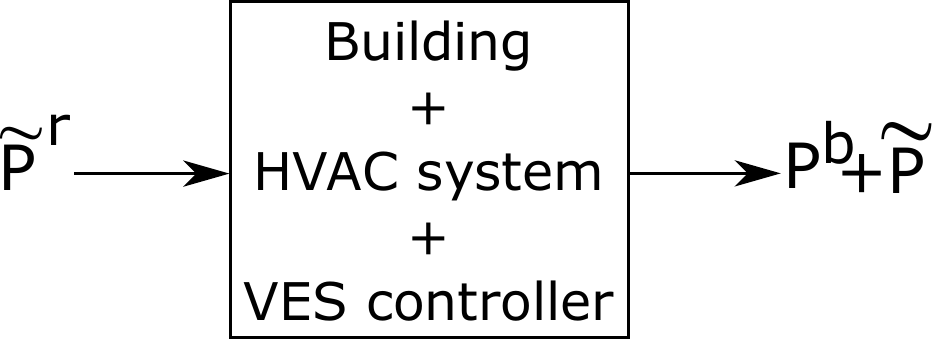}
	\caption{VES system; we assume that the VES controller provides perfect tracking so that $\tilde P(t)$ tracks $\tilde P^r(t)$.}
	\label{fig:VEScontroller}
\end{figure}
\fi

\subsection{Thermal dynamics of HVAC-based VES}
A commonly used modeling paradigm for dynamics of temperature  is resistor-capacitor (RC)
networks~\cite{kramer:2012simplified}. The following simple RC network model is used to model the
temperature of the zone serviced by the HVAC system:
\begin{linenomath*}
\begin{align}\label{eq:T}
C\dot{T}(t) &=  \frac{1}{R}(T_{oa}(t) - T(t)) + q_x(t)+q_\hvac(t),
\end{align}
\end{linenomath*}
where $R$ is the building structure's resistance to heat exchange
between indoors and outdoors, $C$ is the thermal capacitance of the
building, $T_{oa}$ is the outdoor air temperature, $q_x$ is the exogenous heat influx into the building, and
$q_\hvac$ is the heat influx due to the HVAC system, which is due to the
temperature of the air supplied to the building and the air removed from
the zone:
\begin{align}
  \label{eq:q_ac}
  q_\hvac(t) & = m_{a}(t)C_{pa}[T_{sa}(t) - T(t)],
\end{align}
where $m_a$ is the supply air flow rate, $C_{pa}$ is the specific heat capacity of air at constant pressure, $T_{sa}$ is the temperature of the supply air, and $T$ is the temperature of the air leaving the zone. Some of the air leaving the zone is recirculated while some exit the building; see Figure \ref{fig:AHU-chiller-schematic}. Although much more complex models are possible, the simplified model~\eqref{eq:T} aids analysis. Furthermore, it is argued in~\cite{Fux:2014} that a first-order RC network model---such as \eqref{eq:T}---is adequate for prediction up to a few days.

\subsection{HVAC power consumption model}
The power consumption of the HVAC system is a sum of the fan power and
chiller power: $P_\hvac(t)=P_{fan}(t)+P_{ch}(t)$. 

We model the fan power consumption as:
\begin{align}\label{eq:fanpower}
P_f(t) = \alpha_{1f} m_a^2(t)+\alpha_{2f}m_a(t), 
\end{align}
where $\alpha_{1f}(>0)$ and $\alpha_{2f}$ are coefficients that depend on
the fan. Variable speed air supply fan power models reported in the
literature are typically cubic~\cite{rouheiforpib:2001}. We use a quadratic
model for two main reasons. One is ease of analysis, which will be
utilized in Section~\ref{sec:zero-mean-powerref}. The other is that a quadratic
model is adequate to fit measured data, which we will show in
Section~\ref{ssec:single-period}. Note that $\alpha_{2f}$ is allowed to
be negative to better fit measurements, though $P_f$ is always
non-negative for the range of airflows in which we consider the VES system to
be operating. 

Electrical power consumption by the chiller, $P_{ch}$, is modeled as
being proportional to the heat it extracts from the mixed air stream
that passes through the evaporator (or the cooling coil in a chilled
water system): 
\begin{align}\label{eq:Pch-1}
P_{ch}(t) = \frac{m_a(t)[h_{ma}(t)- h_{sa}(t)]}{COP}, 
\end{align}
where $COP$ is the coefficient
of performance of the chiller, $h(\cdot)$ is specific enthalpy of air, and the subscripts $ma$ and $sa$
stand for ``mixed air'' and ``supply air''; see
Figure~\ref{fig:AHU-chiller-schematic}. Since a part of the return air is mixed with the outside air, the specific enthalpy of the mixed air is:
\begin{align}\label{eq:Tma}
	h_{ma}(t)=r_{oa}(t) h_{oa}(t)+(1-r_{oa}(t))h(t),
\end{align}
where $r_{oa}$ is the
so-called outside air ratio: $r_{oa}\eqdef \frac{m_{oa}}{m_{a}}$, $h_{oa}$ is the specific enthalpy of outdoor air, and $h$ is the specific enthalpy of the air leaving the zone. The specific enthalpy of moist air with temperature $T$ and humidity ratio $W$ is given by:
	$h(T,W) = C_{pa}T + W(g_{H_20}+C_{pw}T)$,
where $g_{H_20}$ is the heat of evaporation of water at 0\degree C, and $C_{pa},C_{pw}$ are specific heat of air and water at constant pressure. We assume the following throughout the paper to simplify the analysis:    
\begin{assumption} \label{as:basic}
(i) The ambient temperature ($T_{oa}$), the exogenous heat gain ($q_x$), and the coefficient of performance of the chiller (COP) are constants. (ii) The ambient is warmer than the maximum allowable indoor
          temperature: $T_{oa}>T_H$, so the HVAC system only provides cooling. 
(iii) Effect of humidity change is ignored so that the
          specific enthalpy of an air stream, $h$, with
          (dry-bulb) temperature $T$ is $h=C_{pa}T$, where $C_{pa}$ is
          the specific heat capacity of dry air.
	(iv) The supply air temperature, $T_{sa}$, is constant, and $T_{sa}< T^\b$. 
\end{assumption}
The first three are taken for the ease of analysis. The fourth usually
holds in practice because the cooling coil control loop maintains
$T_{sa}$ at a constant set point, which is lower than indoor
temperature in cooling applications.

With Assumption~\ref{as:basic}, \eqref{eq:fanpower}, \eqref{eq:Pch-1}, and
\eqref{eq:Tma} yield
\begin{align}
	P_\hvac(m_a,T)&=\alpha_{1f}(m_a)^2+\alpha_{2f}m_a \notag \\
				 &+\frac{m_aC_{pa}[r_{oa}T_{oa}+(1-r_{oa})T-T_{sa}]}{COP}.\label{eq:P-hvac-final}
\end{align}
Similarly, the temperature dynamics~\eqref{eq:T} and \eqref{eq:q_ac} become	
\begin{linenomath*}
	\begin{align}\label{eq:T-dynamics-final}
		\dot{T}= \frac{1}{RC}(T_{oa} - T)+ \frac{1}{C}q_x  +\frac{1}{C}m_aC_{pa}(T_{sa}-T).
	\end{align}    
	\end{linenomath*}

        \begin{definition}[Baseline]\label{def:baseline}
          Baseline corresponds to an equilibrium condition in which zone
          temperature and air flow rate are held at constant values, denoted by   $T^\b$ and
          $m_a^\b$.
        \end{definition}
        It follows from Definition~\ref{def:baseline} and
        \eqref{eq:T-dynamics-final} that the baseline variables  $T^\b$ and
          $m_a^\b$ must satisfy
	\begin{linenomath*}
	\begin{align}\label{eq:T-eqlbm-eq}
		0 = \frac{1}{R}(T_{oa}-T^\b) + q_x + m_a^\b C_{pa}(T_{sa}-T^\b).
	\end{align}
	\end{linenomath*}
The \emph{baseline power consumption}, $P_\hvac^\b$, is obtained by plugging in $T^\b$ and $m_a^\b$
into the expression for $P_\hvac$ in \eqref{eq:P-hvac-final}.

The baseline temperature is best thought of as the setpoint that the climate controller uses, and can be any temperature that is strictly inside the allowable interval, meaning $T_L < T^{(b)} < T_H$. Since some variation of the temperature around the setpoint is inevitable due to imperfect reference tracking by a climate controller, the setpoint is always chosen to be inside the allowable limits. The requirement $T_L < T^{(b)} < T_H$ is consistent with this practice. 
 
\subsection{VES system dynamics and power consumption} \label{ssec:diff-algb}
Now we will derive the expressions for the VES system dynamics and power consumption which will be used in the subsequent analysis presented in Section~\ref{sec:zero-mean-powerref}. Let $\tilde m_a(t)$ be the airflow rate deviation (from the baseline)
commanded by the VES controller. Note that $m_a(t) = m_a^\b+ \tilde
m_a(t)$. Let the resulting deviation in the zone temperature be
\begin{align}\label{eq:def-deltaT}
\tilde T(t) \eqdef T(t)  - T^\b . 
\end{align} 
The power consumption by the virtual battery is:
\begin{align}\label{eq:Ptilde-def}
 \tilde P(t) \eqdef & P_\hvac(m_a(t),T(t))- P_\hvac^\b(m_a^\b,T^\b),
\end{align}
where $P_\hvac(\cdot,\cdot)$ is given by \eqref{eq:P-hvac-final}.

By expanding \eqref{eq:Ptilde-def}, we obtain:
	\begin{linenomath*}
	\begin{align}\label{eq:deltaP_nonlinear}
 \tilde P =a\tilde m_a+b\tilde T+c\tilde m_a \tilde T+d\tilde m_a^2,
	\end{align}
	\end{linenomath*}
where the constants $a,b,c,$ and $d$ are:
	\begin{linenomath*}
	\begin{align} 
	a&\eqdef 2\alpha_{1f}m_a^\b +\alpha_{2f} +\frac{C_{pa}[r_{oa}T_{oa}+(1-r_{oa})T^\b-T_{sa}]}{COP}, \label{eq:deltaP_nonlinear_a}\\ 
	b&\eqdef \frac{C_{pa}m_a^\b(1-r_{oa})}{COP},
	c\eqdef \frac{C_{pa}(1-r_{oa})}{COP}, 
	d\eqdef \alpha_{1f}. \label{eq:deltaP_nonlinear_bcd}
	\end{align}
	\end{linenomath*}
Differentiating \eqref{eq:def-deltaT}, and using
\eqref{eq:T-dynamics-final} and \eqref{eq:T-eqlbm-eq} 
we obtain:
	\begin{linenomath*}
	\begin{align}
		\dot{\tilde T}&=-\alpha\tilde T-\beta\tilde m_a-\gamma\tilde T\tilde m_a, \quad \text{ 	where } \label{eq:deltaT-dynamics}\\
          \alpha & \eqdef \frac{RC_{pa}m_a^\b+1}{RC}, 
\beta  \eqdef \frac{C_{pa}(T^\b-T_{sa})}{C}, 
          \gamma \eqdef \frac{C_{pa}}{C}. \label{eq:def-alpha-beta-gamma}
        \end{align}
	\end{linenomath*}
 The dynamics of the temperature deviation (and therefore of the SoC of
 the virtual battery, cf. Definition~\ref{def:soc-vb}) are thus a differential
algebraic equation (DAE):
$   \dot{\tilde T} = f(\tilde T,\tilde m_a)$ , $  \tilde P = g(\tilde T,\tilde m_a)$, where the first (differential) equation is given by \eqref{eq:deltaT-dynamics} and the second (algebraic) equation is given by \eqref{eq:deltaP_nonlinear}. 

\section{Analysis}\label{sec:zero-mean-powerref}
In this paper we restrict the power consumption of the virtual battery to a square-wave signal. There are three reasons for this choice. One, it enables comparison with prior work~\cite{beihisbac:2015,linmatjohhisbac:2017}. Two, it aids the analysis of temperature dynamics. Three, an arbitrary square-integrable signal can be approximated by a combination of square waves using the Haar wavelet transform~\cite{mallat:2008wavelet}.

Let the amplitude of the power consumption $\tilde P(t)$ be $\Delta P$ and the
half-period be $t_p$ (so that the period is $2t_p$). For half of the
period, $\tilde P(t)=\Delta P$, and for the other half, $\tilde
P(t)=-\Delta P$. Consider a complete charge-discharge interval of the
VES system, $[0, \tau]$, so that $SoC(0) = SoC(\tau)$; cf. Definition~\ref{def:cd}. Let
$t_c$ be the total length of the time intervals during which the VES was
charging, i.e., the value of $\tilde{P}(t)$ is $\Delta P$ at any $t$
in those intervals. Similarly, let $t_d$ be the total length of the
time intervals during which the VES was discharging, i.e., the value
of $\tilde{P}(t)$ is $-\Delta P$ at any $t$ in those intervals. Note that $t_c +
t_d = \tau$. It follows from \eqref{eq:rte-def-vb} that
\begin{linenomath*}
\begin{align}\label{eq:rte-simple}
	\eta_\rte =	\frac{-\int_{t_d}[-\Delta{P}]dt}{\int_{t_c}[\Delta{P}]dt}=\frac{t_d}{t_c}.
\end{align}
\end{linenomath*}
The RTE will therefore be either larger or smaller than one depending
on whether $t_d \geq t_c$ or vice versa. The formula~\eqref{eq:rte-simple} will be used in
the subsequent analysis.

Since~\cite{beihisbac:2015} reported differences in observed RTE depending on whether the power consumption is first increased and then decreased from the baseline (``up/down'' scenario), or vice versa (``down/up'' scenario), we treat them separately.

\subsection{A single period of square-wave power consumption} \label{ssec:single-period}
In this section we consider a single period of square-wave power deviation signal. In the ``up/down'' scenario, there are two possibilities for the temperature deviation. The first possibility, which is shown in Figure~\ref{fig:charge-then-discharge-cases-combined}, is that the temperature deviation $\tilde{T}$ is above 0 at the end of one period of the square wave. This means additional charging is needed to bring $\tilde T$ to 0 or alternatively to bring the SoC back to its starting value, which makes the time interval $[0,2t_p+t_{recov1}]$ a complete charge-discharge interval according to Definition~\ref{def:cd}. The RTE computed over this interval using Definition \ref{def:rte} or equivalently \eqref{eq:rte-simple} is called \emph{the RTE for one cycle}. So for the first possibility $t_c = t_p + t_{recov1}$ and $t_d=t_p$, and \eqref{eq:rte-simple} tells us that $\eta_\rte < 1$. The second possibility is that the temperature deviation $\tilde{T}$ is below 0 at the end of one period of the square wave. This means additional discharging is needed to bring  $\tilde{T}$ to 0, which makes the time interval $[0,2t_p+t_{recov2}]$ a complete charge-discharge interval.  Therefore, $t_c = t_p$ and $t_d=t_p+t_{recov2}$, and \eqref{eq:rte-simple} tells us that $\eta_\rte > 1$.

\begin{figure}[h]
	\centering
	\includegraphics[width=0.75\linewidth]{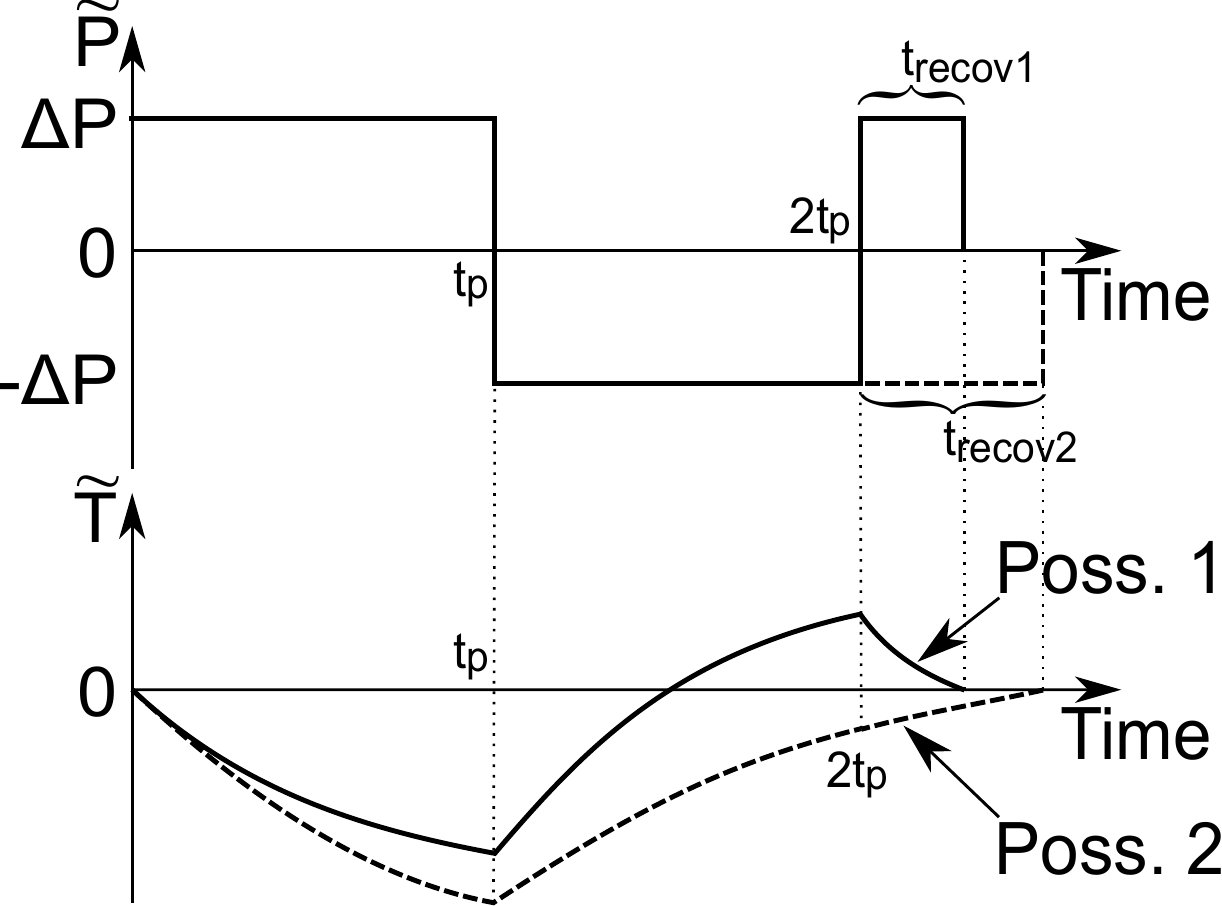}
	\caption{Up/down scenario, possibility 1: since $\tilde T(2t_p)>0$ additional charging is needed to bring back $\tilde T$ to its initial value(=0). Possibility 2: since $\tilde T(2t_p)<0$ additional discharging is needed to bring back $\tilde T$ to its initial value(=0).}\label{fig:charge-then-discharge-cases-combined}
\end{figure}

The situation in the ``down/up'' scenario is similar. The RTE will be smaller or larger than 1 depending on whether the temperature deviation in the first half period is larger or smaller (in magnitude) than that in the second half period. 
\ifArxivVersion
These two possibilities are shown in Figure~\ref{fig:discharge-then-charge-cases-combined}.
\fi

\ifArxivVersion
\begin{figure}[h]
	\centering
	\includegraphics[width=0.75\linewidth]{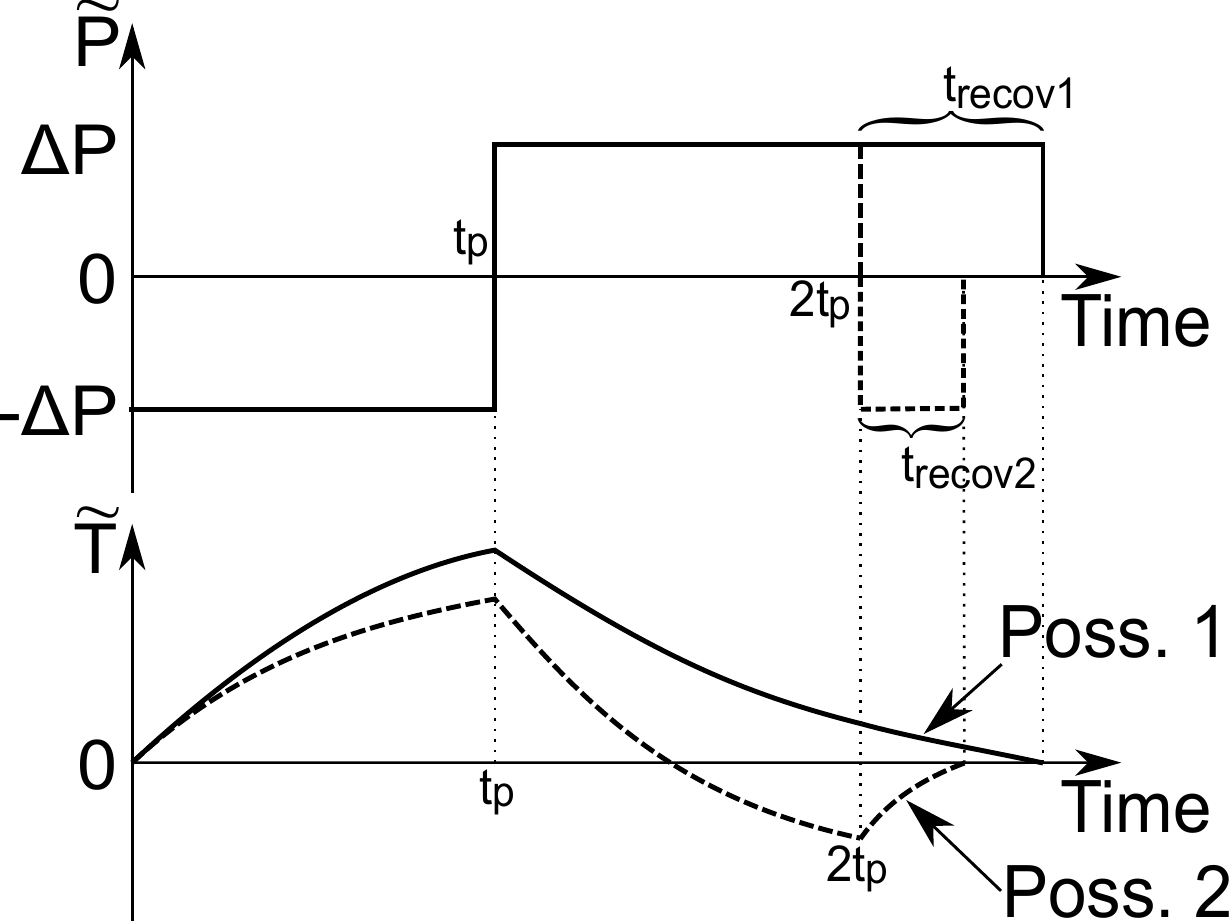}
	\caption{Down/up scenario, possibility 1: since $\tilde T(2t_p)>0$ additional charging is
          needed to bring back $\tilde T$ to its initial
          value (=0). Possibility 2: since $\tilde T(2t_p)<0$ additional discharging is
          needed to bring back $\tilde T$ to its initial
          value (=0).}\label{fig:discharge-then-charge-cases-combined}
\end{figure}
\fi

Lemma \ref{lem:rte-singlecycle} answers the question of which of the possibilities will occur in each
scenario. The proof of the lemma is included in the Appendix. We first state a technical result---Proposition \ref{prop:properties}---that is needed for both
stating and proving the lemma. 
\ifArxivVersion
The proof of the proposition is also included in the Appendix.
\else
The proof of the proposition involves simple algebra and can be found in \cite{raman:2018round}.
\fi

\begin{proposition}\label{prop:properties}
If $r_{oa}=1$ and $\Delta P<P_\hvac^\b$, the following statements hold.
\begin{sublist}
  \item \label{propitem:m-c-d} The airflow rate deviation during charging and discharging are
\begin{linenomath*}
	\begin{align}
	\tilde m_a = 
	\begin{cases}
	 \dfrac{-a + \sqrt{a^2+4d\Delta P}}{2d} \defeq \Delta m_c\; \text{(charging)}, \\ \\
	\dfrac{-a + \sqrt{a^2-4d\Delta P}}{2d} \defeq -\Delta m_d\; \text{(discharging)},
	\end{cases}
	\end{align}
\end{linenomath*}
which satisfy $\Delta m_d >  \Delta m_c > 0$. 
\item \label{propitem:alpha-inequality}  $\alpha >  \gamma \Delta m_d >\gamma
  \Delta m_{c}$.
\item \label{propitem:T-ss} Suppose charging or discharging occurs for infinite time, i.e.,
either $\tilde{P}(t) = \Delta P$ for all $t$ or  $\tilde{P}(t) =-\Delta P$ for all $t$, and let $\tilde T_c^{ss}$, $\tilde T_d^{ss}$
denote the corresponding steady-state values of the temperature deviation $\tilde{T}(t)$. Then $\tilde T_c^{ss}\eqdef\frac{-\beta \Delta m_{c}}{(\alpha+\gamma \Delta m_{c})}<0$ and $\tilde T_d^{ss}\eqdef\frac{\beta \Delta m_d}{(\alpha-\gamma \Delta m_d)}>0$ irrespective of the initial condition $\tilde{T}(0)$, and $\abs{\tilde T_c^{ss}} < \abs{\tilde T_d^{ss}}$.
  \end{sublist}
\end{proposition}

Now we are ready to state the lemma.

\begin{lemma}\label{lem:rte-singlecycle}
Suppose $r_{oa}=1$ (i.e., 100\% outside air) and the time period
$2t_p$ is small enough so that $(\alpha +\gamma \Delta m_c)t_p \ll 1$ ($\alpha$, $\gamma$ are
defined in \eqref{eq:def-alpha-beta-gamma} and $\Delta m_c$ is defined in Proposition~\ref{prop:properties}\ref{propitem:m-c-d}), which implies that the approximation $e^x \approx 1+x$
  is accurate with $x$ replaced by $(\alpha +\gamma \Delta m_c)t_p$. Then, in the up/down
  scenario, the RTE for one cycle is $\eta_\rte<1$ (possibility 1 shown in Figure~\ref{fig:charge-then-discharge-cases-combined}). In down/up scenario, there is a
  critical value $t_p^*$:
    \begin{align}
  \label{eq:def-tp*}
t_p^* \eqdef \frac{-1}{\alpha+\gamma \Delta m_c}\log \frac{\Delta m_c}{\Delta m_d}, 
\end{align} 
 such that if $t_p < t_p^*$, then
 \ifArxivVersion
  $\eta_\rte<1$ for one cycle (possibility 1 shown in Figure \ref{fig:discharge-then-charge-cases-combined}); otherwise $\eta_\rte>1$ (possibility 2 shown in Figure~\ref{fig:discharge-then-charge-cases-combined}).
  \else
  $\eta_\rte<1$ for one cycle; otherwise $\eta_\rte>1$.
  \fi
\end{lemma}
\begin{comment}\label{com:single-cycle}
  The RTE values obtained in the LANL experiments are almost always less than 1, in both up/down
  and down/up scenarios, but in a small fraction of up/down and
  down/up experiments the RTE was observed to be larger than one; see
  Figure 5 of \cite{beihisbac:2015}. While Lemma \ref{lem:rte-singlecycle} shows that it is possible for the RTE
  to be either larger or smaller than 1 as observed in the experiments, its prediction that $\eta_\rte$ cannot be greater than 1 for the up/down scenario is inconsistent with the
  observation in~\cite{beihisbac:2015}. Interestingly, the simulation
  study~\cite{linmatjohhisbac:2017} also observed that the RTE is
  smaller than 1 for the up/down scenario and greater than 1 for
  down/up scenario. This is consistent with our results but inconsistent with LANL experiments. In~\cite{linmatjohhisbac:2017}, they did not test for small enough values for the time period to notice the existence of a critical time period in the down/up scenario. 
\end{comment}

The assumptions made in the lemma are for ease of analysis; its predictions still hold when they are violated. Figure~\ref{fig:RTE-robustness-tp} shows the numerically computed $\eta_\rte$ for various values of $t_p$ using the parameter values listed in the next paragraph. We see from the Figure~\ref{fig:RTE-robustness-tp} that the predictions regarding $\eta_\rte$ from Lemma~\ref{lem:rte-singlecycle} hold even when $(\alpha+\gamma \Delta m_c)t_p$ is not small and $r_{oa}$ is not 1. For instance, when $t_p=300$ minutes, $(\alpha+\gamma \Delta m_c)t_p = 1.8$, which is not tiny; yet numerically computed values are consistent with the lemma's prediction.

%
The following parameters were chosen for the numerical computations:
$T_{sa}=55$\degree F, $T^\b=72$\degree F, $T_L=70$\degree F,
and $T_H=74$\degree F.
\ifArxivVersion
 \begin{figure}[htpb]
 	\centering
 	\includegraphics[width=0.98\linewidth]{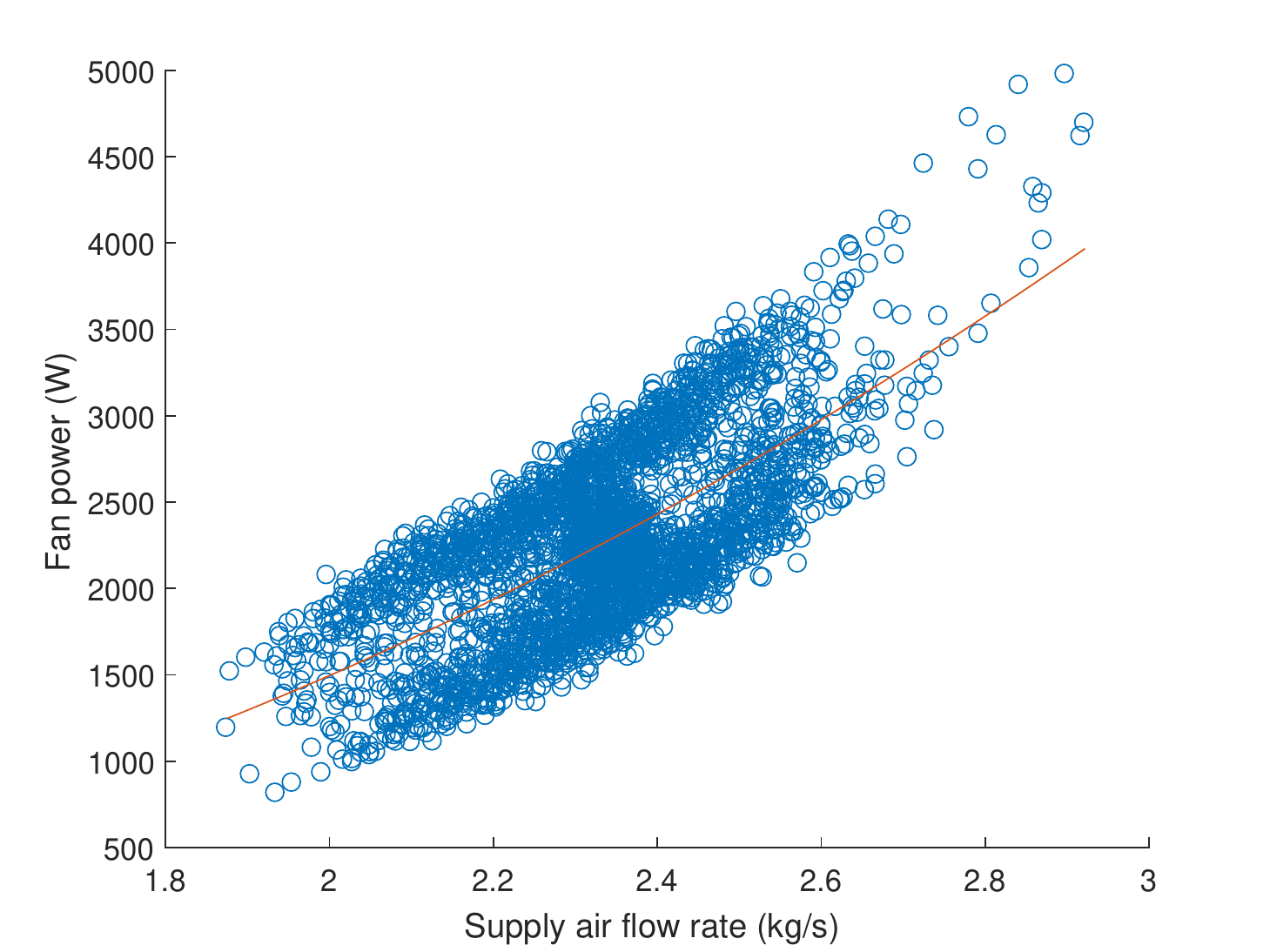}
 	\caption{Fan power vs. airflow rate; measurements from AHU-2
          of Pugh Hall at UF (circles), and predictions from the
          best fit model \eqref{eq:fanpower} to the measurements (curve).}
 	\label{fig:fanpower-ahu2-quadratic}
 \end{figure}
\fi
The building in this paper is based on a large auditorium ($\sim 6$ m high, floor area of $\sim 465 $ m$^2$) in Pugh Hall located in the University of Florida campus, which is served by a dedicated air handling unit. We choose $m_a^\b=2.27$ kg/s since that is representative of the airflow rate to this zone. We choose the following parameters, guided by~\cite{LinThesis:2014}: $C=3.4\times10^7$ J/K and $R=1.3\times10^{-3}$ K/W. We also choose $T_{oa}=80^\circ$F and $COP=3.5$, somewhat arbitrarily. 
The fan power coefficients were chosen to be $\alpha_{1f}=662$ W/(kg/s)$^2$ and $\alpha_{2f}=-576$ W/(kg/s), based on fitting a quadratic model to measured fan power from the zone in  
\ifArxivVersion
question; see Figure~\ref{fig:fanpower-ahu2-quadratic}.
\else
question. The data and the fit can be found in Figure~5 of \cite{raman:2018round}.
\fi

\begin{figure}[h]
	\centering
	\includegraphics[width=0.98\linewidth]{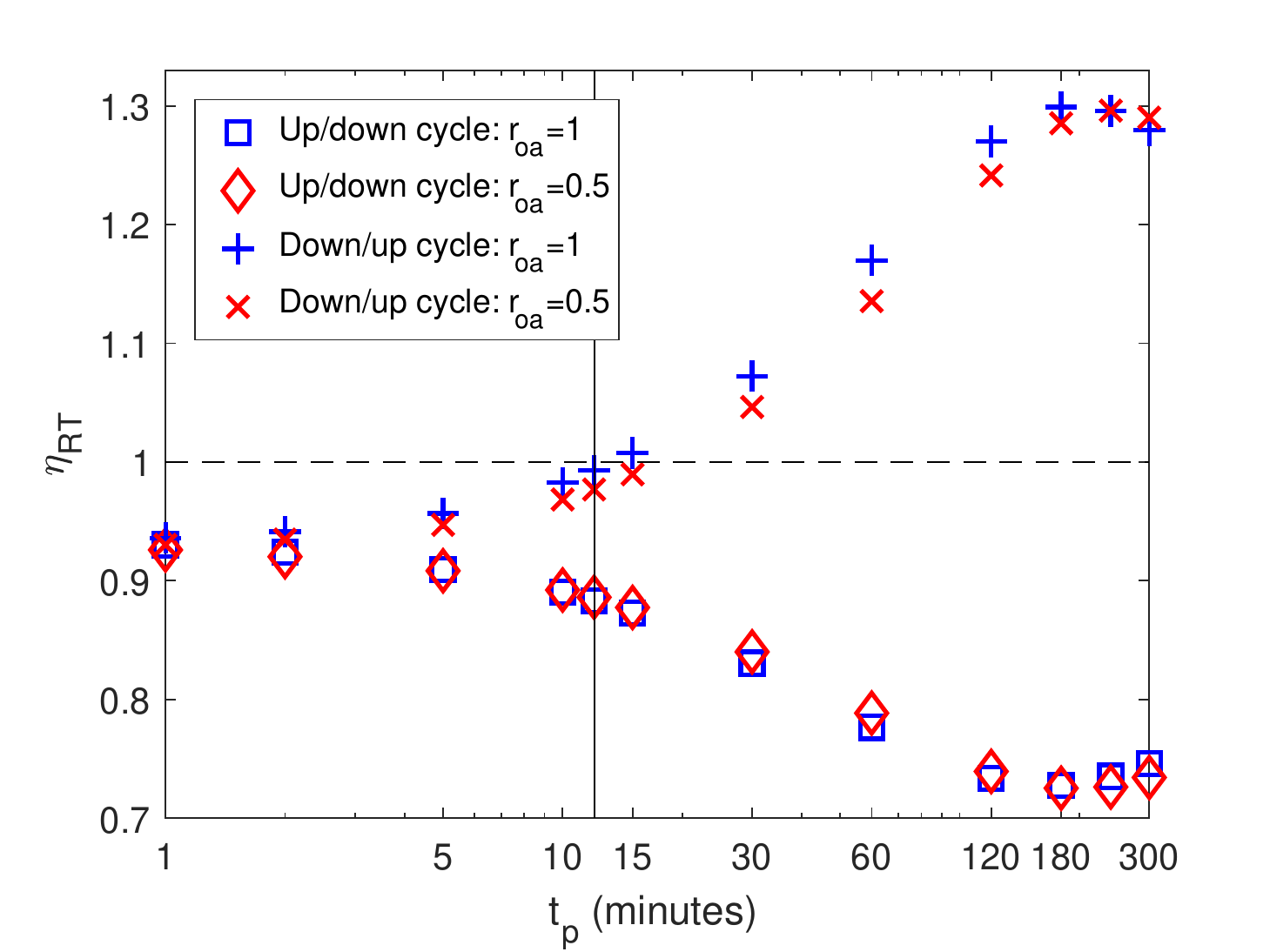}
	\caption{$\eta_\rte$ vs. $t_p$, for $r_{oa}=1$ and $r_{oa}=0.5$; $\Delta P = 0.2P^\b_\hvac$. The vertical line shown is $t_p^*$ ($\approx$12 minutes) computed from \eqref{eq:def-tp*} for $r_{oa}=1$.}
	\label{fig:RTE-robustness-tp}
\end{figure}

\subsection{Multiple periods of square-wave power consumption}
We now consider $n$ periods of the square-wave, $n>1$. At the end of
$n$ periods, the temperature deviation may not be exactly 0 (i.e.,
$\tilde T(n2t_p) \neq 0$), even though its initial value was $0$
(i.e., $\tilde T(0)=0$). Charging or discharging might be needed for
an additional amount of time $t_{recov}$ to bring the temperature deviation back
to $0$. Whether recovery to the initial SoC requires additional
charging or additional discharging depends on whether $\tilde T(n2t_p)$ is positive or
negative. In either case, since $\tilde{T}(0)= \tilde{T}(n2t_p+t_{recov})=0$, according to
Definition~\ref{def:cd}, the time interval $[0, \; n2t_p+t_{recov}]$
constitutes a complete charge-discharge interval of the virtual battery. The RTE computed over this interval using Definition \ref{def:rte} or equivalently \eqref{eq:rte-simple} is called \emph{the RTE for n cycles} or $\eta_\rte(n)$.

Figure~\ref{fig:bringSOCto0_mult} shows an illustration of the two
possible scenarios for the possible values of $\tilde{T}(n2t_p)$. For
the sake of concreteness, we have assumed the VES service starts with
a down/up cycle in the figure. 
\begin{figure}[h]
		\centering
		\includegraphics[width=0.75\linewidth]{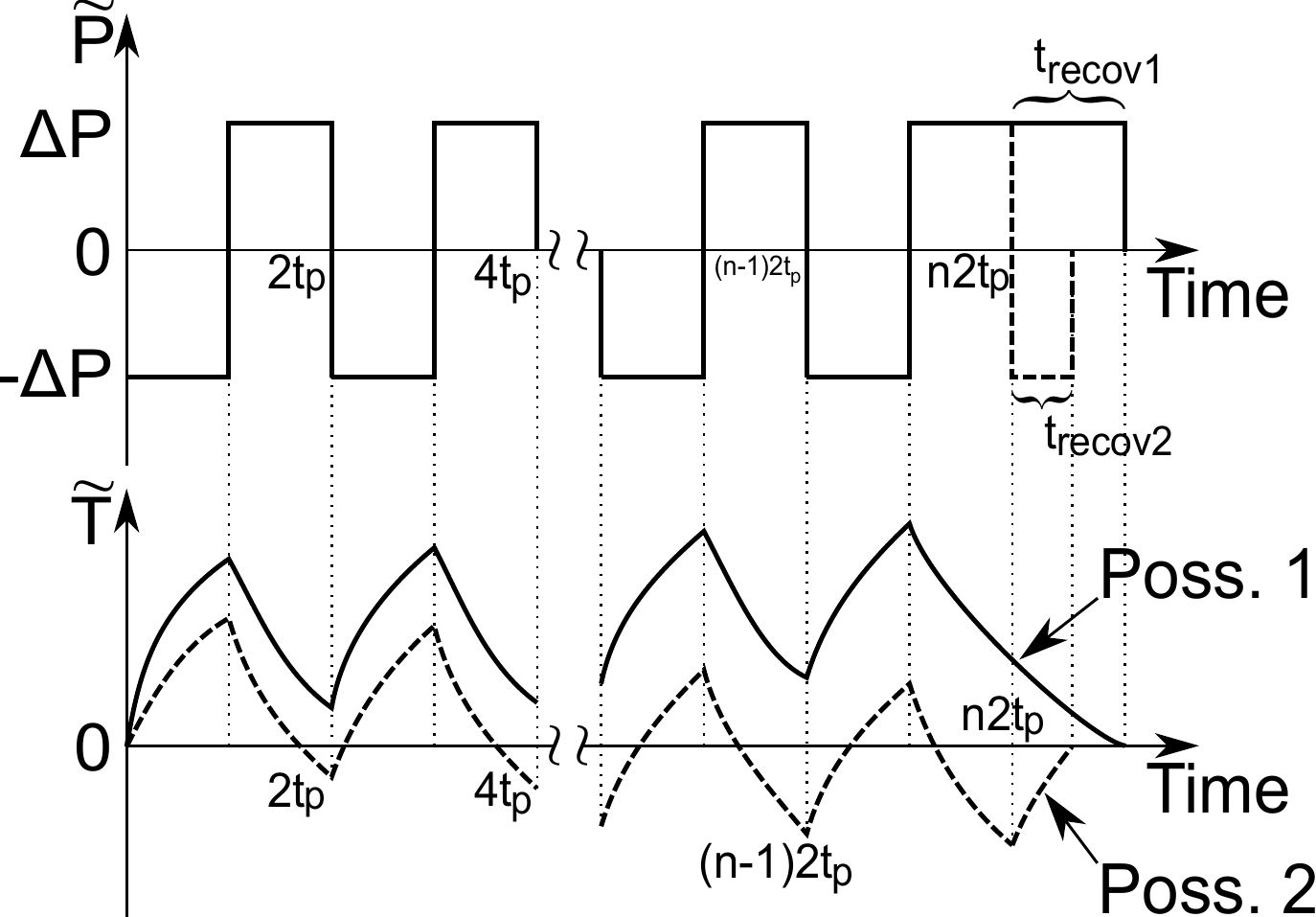}
		\caption{Additional charging or discharging needed to
                  bring $\tilde T$ to its initial value (=0)
                  after $n$ periods of down/up cycle.}\label{fig:bringSOCto0_mult}
\end{figure}
In the first possibility, denoted by the solid lines,
$\tilde{T}(n2t_p)\geq 0$, and therefore
additional charging is performed for $t_{recov1}\geq 0$ amount of time in
order to bring the temperature deviation to 0. If this possibility
were to occur, within
the complete charge-discharge interval of $[0,n2t_p+t_{recov1}]$, the
charging time is $nt_p+t_{recov1}$, while the discharging time is
$nt_p$. It now follows from \eqref{eq:rte-simple} that for possibility 1
\begin{align}\label{eq:eta-n-cd}
  \eta_\rte(n) = \frac{t_d}{t_c} = \frac{n t_p}{n t_p+t_{recov1}(n)} \leq 1.
\end{align}
In the second possibility, denoted by the dashed lines, $\tilde{T}(n2t_p)
\leq  0$ and therefore
additional discharging is needed for $t_{recov2}\geq 0$ amount of
time. For this possibility, 
\begin{align}\label{eq:eta-n-dc}
  \eta_\rte(n) = \frac{t_d}{t_c} = \frac{n t_p+t_{recov2}(n)}{n t_p} \geq 1.
\end{align}
If the VES service were to start with an up/down cycle, the same two
possibilities exist in principle, so again the RTE can be smaller or larger than one
depending on whether the temperature deviation at the end of the $n$
periods is positive or negative. 

The proof of the main result of the paper, Theorem \ref{lem:rte-asymptotic}, needs a key intermediate result which is presented in the next lemma.

\begin{lemma}\label{prop:T-bounded}
	For $r_{oa}=1$ and $\tilde T(0)=0$, the magnitude of the temperature deviation $|\tilde{T}(t)|$ is bounded by $max\Big\{ \lvert\tilde T_c^{ss}\rvert, |\tilde T_d^{ss}|\Big\}, \, \forall t$ ($\tilde T_c^{ss}$ and $\tilde T_d^{ss}$ are defined in Proposition~\ref{prop:properties}\ref{propitem:T-ss}).
\end{lemma}

The proof of Lemma \ref{prop:T-bounded} is presented in the Appendix.

\begin{theorem}\label{lem:rte-asymptotic}
If $r_{oa}=1$, $\lim_{n \to \infty} \eta_\rte(n) = 1$.  
\end{theorem}

\begin{proof}[Proof of Theorem \ref{lem:rte-asymptotic}]
Consider the first possibility:
$\tilde{T}(n2t_p)\geq0$ so that additional charging is needed for some
time, and call that time $t_{recov}(n) \geq 0$. From~\eqref{eq:eta-n-cd}, we have
\begin{align*}
\lim_{n \to \infty}  \eta_\rte(n) = \lim_{n \to \infty} \frac{n t_p}{n t_p+t_{recov}(n)} =  \lim_{n \to \infty} \frac{t_p}{t_p+\frac{1}{n} t_{recov}(n)}.
\end{align*}
Since the building is undergoing charging for $t > n2t_p$, it follows from Proposition~\ref{prop:properties}\ref{propitem:T-ss} that the temperature deviation monotonically decays toward the value $\tilde T_c^{ss}$ from the ``initial value'' $\tilde{T}(n2t_p)$. Since $\tilde{T}(n2t_p)$ is bounded by a constant that is independent of $n$, which follows from Lemma~\ref{prop:T-bounded}, the time it takes for $\tilde{T}(t)$ to reach 0 from its ``initial value'' $\tilde{T}(n2t_p)$ is upper bounded by a constant independent of $n$,
which we denote by $\overline t_{recov}$. Thus, $t_{recov}(n) \leq
\bar{t}_{recov}$, and $\bar{t}_{recov}$ is a constant independent
of $n$. Therefore, $\lim_{n \to \infty }\frac{1}{n}t_{recov}(n) = 0$,
and therefore, $ \lim_{n  \to \infty} \eta_\rte(n) = 1$.  A similar
analysis holds for the \
second possibility: $\tilde{T}(n2t_p)\leq0$. In this case additional
discharging is needed for $t > n2t_p$. Again, the time it takes
for the temperature deviation to get back to $0$ is upper bounded by a
constant independent of $n$ since the ``initial condition''
$\tilde{T}(n2t_p)$ is upper bounded (in magnitude) by a constant
independent of $n$. Thus, again $\frac{1}{n} t_{recov}(n) \to 0$ as $n
\to \infty$, and therefore $\lim_{n \to \infty}  \eta_\rte(n)$  = $\lim_{n \to \infty} \frac{n  t_p+t_{recov}(n)}{n t_p} =1$.
\end{proof}


\begin{comment}\label{com:rte-vs-n}
  \begin{sublist}
  \item Theorem~\ref{lem:rte-asymptotic} indicates that it is better for an HVAC-based virtual battery to provide VES services for an extended period of time, so that $n$ is large and hence RTE is close to 1. If VES service is stopped after a small number of cycles, RTE can be significantly lower than 1. This will entail a loss of efficiency and thus an increase in the energy cost for the building owner/operator. 
  
  \item Theorem \ref{lem:rte-asymptotic} holds as long as the temperature deviation $\tilde{T}(t)$ is bounded by a constant independent of $t$, since that alone is sufficient to guarantee that $\frac{1}{n}t_{recov}(n) \to 0 $ as $n \to \infty$. Therefore, the Theorem is robust to the kinds of model (and parameter values in the model) used in the analysis; the asymptotic RTE is 1 as long as the temperature deviation is bounded by a constant. \nsr{Might have to include something regarding idealized models used, for satisfying reviewer 2's final comment.}
    \end{sublist}
\end{comment}

\section{Numerical verification}\label{sec:simulations}
 In order to show that the main result---Theorem~1---is robust to modeling assumptions made during analysis, we test the prediction using a more sophisticated model in simulations that includes humidity. The temperature dynamics are modeled as follows:
 \begin{align*}
 C_z\dot{T}(t) &=  \frac{1}{R_w}(T_{w}(t) - T(t)) + q_x(t)+q_\hvac(t) \\
 C_w\dot{T}_w(t) &= \frac{1}{R_z}(T_{oa}(t) - T_w(t)) + \frac{1}{R_w}(T(t) - T_w(t)) 
 \end{align*}
 where $T_w$ is the wall temperature, $C_z$ and $C_w$ are the thermal capacitance of the zone and the wall respectively, $R_z$ is the resistance to heat exchange between the outdoors and wall, and $R_w$ is the resistance to heat exchange between the wall and indoors. $q_\hvac$ is the heat influx due to the HVAC system which is given by \eqref{eq:q_ac}. The dynamics of zone humidity ratio $W$ is modeled as \cite{SG_PB_Energy:11}: 
 \begin{align*}
 \dot{W}(t) = \frac{R_gT(t)}{VP^{da}}\Bigg[\omega_x(t) + m_a(t)\frac{W_{sa}(t)-W(t)}{1+W_{sa}(t)}\Bigg]
 \end{align*}
 where $V$ is the volume of dry air (which is same as the zone volume), $R_g$ is the specific gas constant of dry air, $P^{da}$ is the partial pressure of dry air, $W_{sa}$ is the supply air humidity ratio, and $\omega_x$ is the rate of internal water vapor generation.  Models \eqref{eq:fanpower} and \eqref{eq:Pch-1} are used to compute the fan and the chiller power respectively. Chiller $COP$ is modeled as a linear function of $T_{oa}$: $COP(t) = 5.5 - 0.025T_{oa}(t)$, with $COP$ saturating at 4 for $T_{oa}\leq 60\degree F$ and 3 for $T_{oa}\geq 100 \degree F$. This model is an approximation of the single-speed electric DX (direct expansion) air cooling coil model from \cite{doe:2018energyplus}.

 The baseline power consumption is computed by performing a simulation with the climate control system. Then we perform the VES simulation with the square-wave power deviation reference added to the baseline power computed, which is provided as a power reference to the VES controller, as described in Section ~\ref{sec:model}. At the end of the ancillary service event the VES controller is turned off and the zone climate controller is turned on to bring the zone temperature to its set point.\nsr{Is this explanation good? Also is we decide on using the name VES+climate controller in the review response, we need to change it here.} Figure~\ref{fig:rte-vs-nofcycles} shows numerically computed values of $\eta_\rte(n)$ as a function of $n$. The RTE was computed using \eqref{eq:rte-def-vb}. The result presented in the figure is consistent with the prediction of Theorem~\ref{lem:rte-asymptotic} that the RTE tends to $1$  as $n \to \infty$.
\begin{figure}[htpb]
	\centering
	\begin{subfigure}{.5\textwidth}
		\centering
		\includegraphics[width=0.8\linewidth]{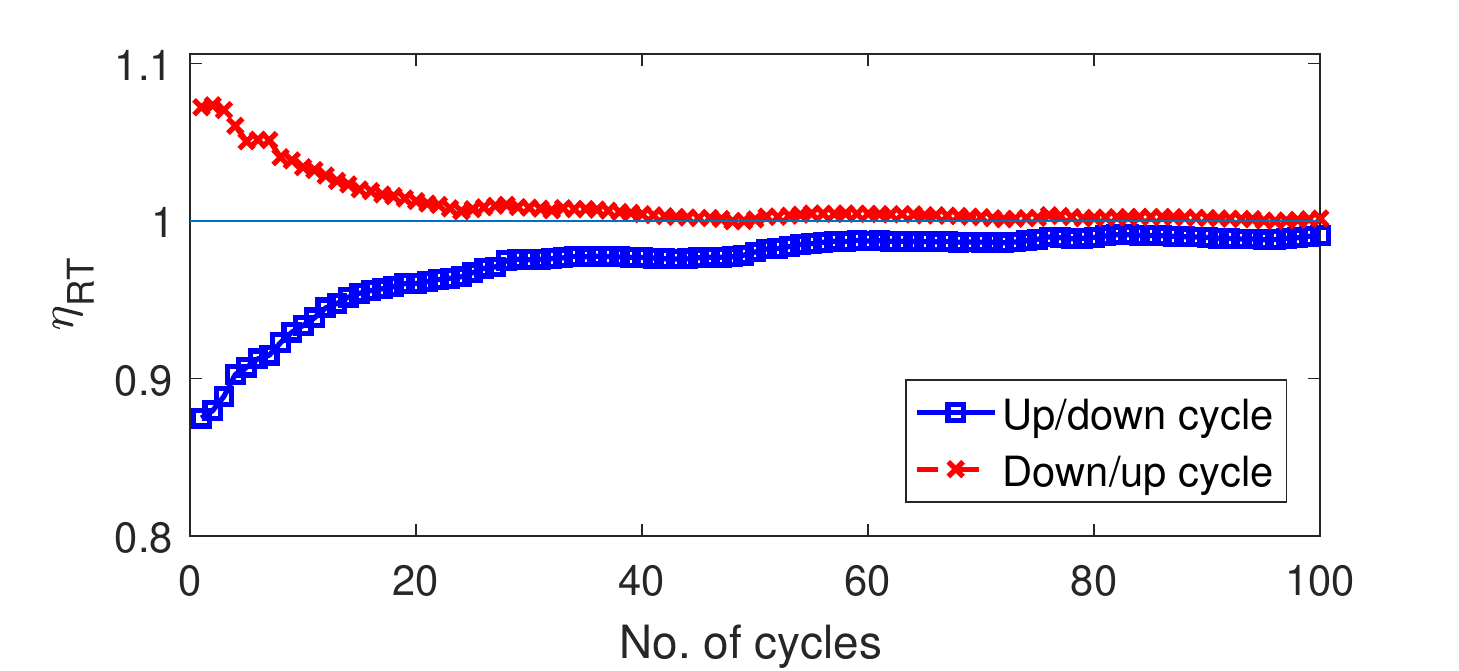}
		\caption{ $\eta_\rte(n)$ vs. $n$, when the  virtual battery tracks
			a square wave power reference ($\Delta P=4500$W, time period $=1$ hour).}
		\label{fig:rte-vs-nofcycles}
	\end{subfigure}
	\begin{subfigure}{.5\textwidth}
		\centering
		\includegraphics[width=0.9\linewidth]{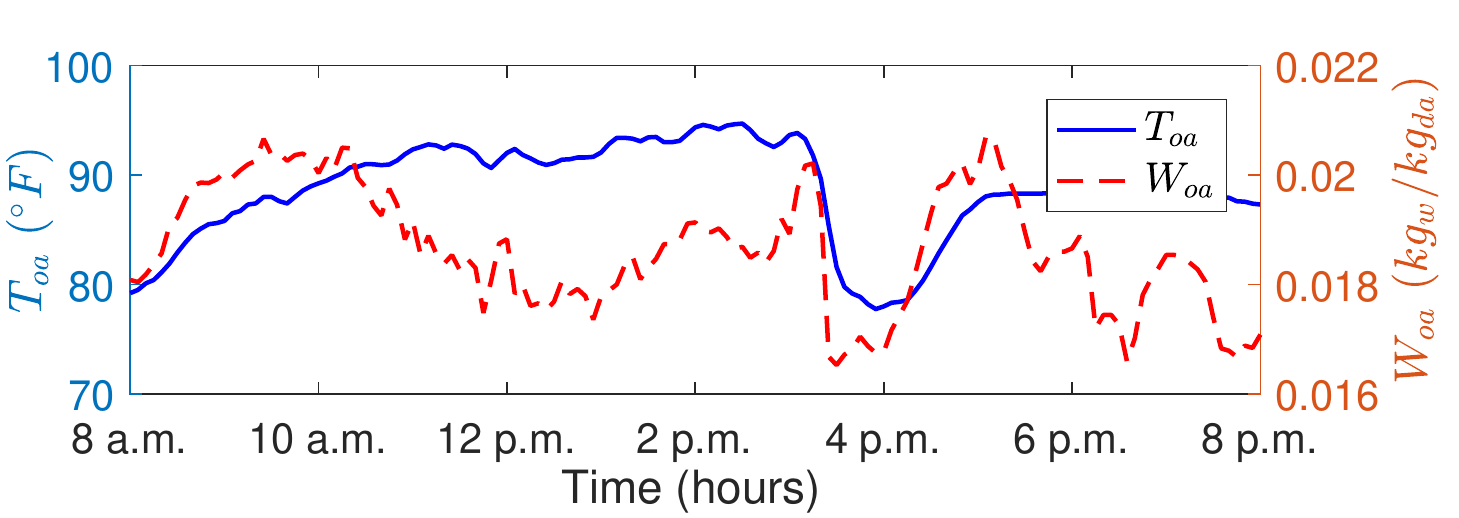}
		\caption{Outside air temperature and humidity ratio used in simulations.}
		\label{fig:Toa_Woa_dayag_plotyy}
	\end{subfigure}
	\caption{Robustness to modeling assumptions.}
	\label{fig:sim_results}
\end{figure}

\section{Conclusion}\label{sec:conclusion}
The main result of the paper is that the asymptotic RTE is unity. It is therefore better for an HVAC-based VES system to be used continuously for a long time than occasionally. The latter can cause low round trip efficiency, while the former has an efficiency close to $1$.

There are several additional avenues for further exploration. The discrepancy between our predictions and results in~\cite{beihisbac:2015} for the ``single demand-response event'' calls for further studies; cf. Comment \ref{com:single-cycle}.  Although the indoor temperature deviation is small (in the sub-1\degree F range) for the range of power deviations examined in our numerical simulations (5\%-30\%), the deviation is not zero mean. This can be interpreted as a slight warming---or cooling---of the building due to VES operation. Examination of the RTE, when the average temperature variation from the baseline is constrained to be 0, is ongoing and some preliminary progress in this direction has been reported in \cite{RamanAnalysisHPB:2018}.

\section*{Acknowledgment}
Prabir Barooah thanks Scott Backhaus for stimulating discussions regarding RTE during a visit to LANL in 2014, and Naren Srivaths Raman thanks Jonathan Brooks for helpful discussions. 


\appendix
\label{app:proofs}

\ifArxivVersion
We start with a technical result first.
\begin{proposition}\label{prop:technical}
	\begin{sublist}
		\item\label{propitem:a-pos} The parameter $a$ defined in \eqref{eq:deltaP_nonlinear_a} is
		positive for every positive $m_a$.
		\item\label{propitem:a-sqr} If $r_{oa}=1$, then $a^2>4d\Delta{P}$ for any feasible $\Delta P$.
		\item\label{propitem:a-d} If $r_{oa}=1$, then $\frac{a}{d}>m_a^\b$.
		\item\label{propitem:m-a} If $r_{oa}=1$ and $\Delta P \leq P_\hvac^\b$, then $m_a^\b \leq \dfrac{1}{2d}(a + \sqrt{a^2-4d\Delta P})$, with equality only if $\Delta P = P_\hvac^\b$.
	\end{sublist}
\end{proposition}

\begin{proof}[Proof of Proposition \ref{prop:technical}]
	\begin{sublist}
		\item It follows from ~\eqref{eq:deltaP_nonlinear_a} that
		\begin{align*}
		a m_a^\b = \alpha_{1f}(m_a^\b)^2+P_\hvac^\b.
		\end{align*}
		Since the right hand side is positive and $m_a^\b>0$, we have that $a>0$.
		\item  For $r_{oa} =1$ it follows from~\eqref{eq:deltaP_nonlinear_a} that:
		\begin{linenomath*}
			\begin{align*}
			a = 2\alpha_{1f}m_a^\b +\alpha_{2f} +\frac{C_{pa}[T_{oa}-T_{sa}]}{COP}.
			\end{align*}
		\end{linenomath*}
		The maximum value that $4d\Delta P$ can take is when $\Delta P =P_\hvac^\b$. Substituting for $\Delta P=P_\hvac^\b$ and from~\eqref{eq:deltaP_nonlinear_bcd} we get:
		\begin{linenomath*}
			\begin{align*}
			4\alpha_{1f}P_\hvac^\b 
			= 4\alpha_{1f}\Bigg[\alpha_{1f}&(m_a^\b)^2+\alpha_{2f}m_a^\b \\
			&+\frac{m_a^\b C_{pa}[T_{oa}-T_{sa}]}{COP}\Bigg].
			\end{align*}
		\end{linenomath*}
		So we need to prove that:
		\begin{linenomath*}
			\begin{align*}
			&\Bigg[2\alpha_{1f}m_a^\b +\alpha_{2f} +\frac{C_{pa}[T_{oa}-T_{sa}]}{COP}\Bigg]^2 > \\
			&~~4\alpha_{1f}\Bigg[\alpha_{1f}(m_a^\b)^2+\alpha_{2f}m_a^\b +\frac{m_a^\b C_{pa}[T_{oa}-T_{sa}]}{COP}\Bigg].
			\end{align*}
		\end{linenomath*}
		This simplifies to
		\begin{linenomath*}
			\begin{align*}
			\Bigg[\alpha_{2f} + \frac{C_{pa}[T_{oa}-T_s]}{COP}\Bigg]^2>0
			\end{align*}
		\end{linenomath*}
		which is always true, and therefore $a^2>4d \Delta P$.
		\item It follows from~\eqref{eq:deltaP_nonlinear_a} and~\eqref{eq:deltaP_nonlinear_bcd} that when $r_{oa}=1$
		\begin{linenomath*}
			\begin{align*}
			\frac{a}{d}=\frac{2\alpha_{1f}m_a^\b +\alpha_{2f} +\frac{C_{pa}[T_{oa}-T_{sa}]}{COP}}{\alpha_{1f}}.
			\end{align*}
		\end{linenomath*}
		With further algebraic manipulation it reduces to $m_a^\b+P_\hvac^\b/(\alpha_{1f}m_a^\b)$. Since $\alpha_{1f}$, $m_a^\b$, and $P_\hvac^\b$ are positive, we have $a/d > m_a^\b$.
		
		\item Let us look at the following expression:
		\begin{align}\label{eq:prop3-proof4-1}
		\frac{a+\sqrt{a^2-4d\Delta P}}{2d}.
		\end{align}
		Substituting for $\Delta P = P_\hvac^\b$ in the above expression and using the expressions for $a$ and $d$ from ~\eqref{eq:deltaP_nonlinear_a} and~\eqref{eq:deltaP_nonlinear_bcd} respectively, we get:
		\begin{linenomath*}
			\begin{align}\label{eq:prop3-proof4-2}
			\frac{2\alpha_{1f}m_a^\b +\alpha_{2f} +\frac{C_{pa}[T_{oa}-T_{sa}]}{COP}+\abs{\alpha_{2f} +\frac{C_{pa}[T_{oa}-T_{sa}]}{COP}}}{2\alpha_{1f}}.
			\end{align}
		\end{linenomath*}
		If $\alpha_{2f} +\frac{C_{pa}[T_{oa}-T_{sa}]}{COP}>0$ then \eqref{eq:prop3-proof4-2} becomes:
		\begin{linenomath*}
			\begin{align*}
			m_a^\b +\frac{\alpha_{2f} +\frac{C_{pa}[T_{oa}-T_{sa}]}{COP}}{\alpha_{1f}},
			\end{align*}
		\end{linenomath*}
		which is greater than $m_a^\b$ since $\alpha_{1f}>0$. If $\alpha_{2f} +\frac{C_{pa}[T_{oa}-T_{sa}]}{COP}<0$, then \eqref{eq:prop3-proof4-2} becomes equal to $m_a^\b$. In \eqref{eq:prop3-proof4-1} for any $\Delta P < P_\hvac^\b$ the value of \eqref{eq:prop3-proof4-1} increases and therefore is greater than $m_a^\b$. This completes the proof.
	\end{sublist}
\end{proof}

\begin{proof}[Proof of Proposition \ref{prop:properties}]
	\begin{sublist}
		\item Since $r_{oa}=1$ from \eqref{eq:deltaP_nonlinear_bcd}, $b=0$ and $c=0$. Therefore the solution for $\tilde m_a$ as a function of $\tilde P$ from \eqref{eq:deltaP_nonlinear} reduces to:
		\begin{linenomath*}
			\begin{align*}
			\tilde m_a = \frac{-a \pm \sqrt{a^2 + 4d\tilde P}}{2d}.
			\end{align*}
		\end{linenomath*}
		During charging $\tilde P = \Delta P$, so the the two roots in the
		equation above are 
		$ \frac{-a + \sqrt{a^2 + 4d\Delta P}}{2d}$ and $\frac{-a - \sqrt{a^2 +
				4d\Delta P}}{2d}$. The second root is not possible, since it is
		negative with a minimum magnitude $\frac{+a+ \sqrt{a^2}}{2d}=\frac{a}{d}$, which is larger than $m_a^\b$ by Proposition~\ref{prop:technical}\ref{propitem:a-d},
		making the total airflow rate negative. Therefore during charging,
		the airflow rate is $\frac{-a + \sqrt{a^2 + 4d\Delta P}}{2d}$. This
		proves the first statement, regarding  $\Delta m_c$. During
		discharging, $\tilde P = -\Delta P$, so the two possible roots are $ \frac{-a + \sqrt{a^2 - 4d\Delta P}}{2d}$ and $\frac{-a - \sqrt{a^2 -
				4d\Delta P}}{2d}$. The second root is not possible, since it is negative with a minimum magnitude larger than $m_a^\b$ for $\Delta P<P_\hvac^\b$ from Proposition~\ref{prop:technical}\ref{propitem:m-a}, which will make the total air flow rate negative. Therefore during charging,
		the airflow rate is $\frac{-a + \sqrt{a^2 - 4d\Delta P}}{2d}$. This proves the second statement, regarding  $\Delta m_d$.
		
		To prove the inequality $\Delta m_d>\Delta m_c$, let $\nu \triangleq
		4d \Delta P$ for simplifying the notation. The inequality $\Delta
		m_d>\Delta m_c$ is equivalent to:
		\begin{linenomath*}
			\begin{align*}
			&\frac{-a + \sqrt{a^2 + \nu}}{2d} < \frac{a - \sqrt{a^2 - \nu}}{2d} \\
			&\Rightarrow -a + \sqrt{a^2 + \nu} < a - \sqrt{a^2 - \nu} \text{, as } d>0.
			\end{align*}
		\end{linenomath*}
		Further algebraic manipulation gives,
		\begin{linenomath*}
			\begin{align}\label{eq:proof_Ma}
			\sqrt{a^2 + \nu} - \sqrt{a^2 - \nu} > \frac{\nu}{a}.
			\end{align}
		\end{linenomath*}
		Since $a^2>\nu$ from Proposition~\ref{prop:technical}\ref{propitem:a-sqr}, let us define $a^2=\nu + \epsilon$ where $\epsilon>0$. Therefore, \eqref{eq:proof_Ma} becomes:
		\begin{linenomath*}
			\begin{align*}
			&\sqrt{2\nu + \epsilon} - \sqrt{\epsilon} > \frac{\nu}{\sqrt{\nu + \epsilon}} \\
			&\Rightarrow \sqrt{(2\nu + \epsilon)(\nu + \epsilon)} > \nu + \sqrt{\epsilon(\nu + \epsilon)}.
			\end{align*}
		\end{linenomath*}
		Squaring on both sides yields:
		\begin{linenomath*}
			\begin{align*}
			\Rightarrow \nu^2 + 2\nu\epsilon > 2\nu\sqrt{\epsilon(\nu+\epsilon)},
			\end{align*}
		\end{linenomath*}
		squaring again on both sides and simplifying, we get: $\nu^4>0$, which
		is true, and therefore $\Delta m_{c}<\Delta m_{d}$.
		\item Note that the maximum value that $\Delta m_{d}$ can take is
		$m_a^\b$; otherwise, the total airflow rate will be negative. For that value of $\Delta m_d$, $\gamma \Delta m_{d} =
		\frac{C_{pa}m_a^\b}{C}$ (as $\gamma = C_{pa}/C$). Substituting for $\alpha$ from \eqref{eq:def-alpha-beta-gamma} and since $R,C>0$, we have
		$\frac{C_{pa}m_a^\b}{C} + \frac{1}{RC} > \frac{C_{pa}m_a^\b}{C}$ so that
		$\alpha > \gamma \Delta m_d$. For the second inequality, note that from Proposition\ref{prop:properties}\ref{propitem:m-c-d} $\Delta
		m_c <\Delta
		m_d$. Since $\gamma$ is positive, $\gamma \Delta m_c<\gamma \Delta m_d$. Therefore, $\alpha>\gamma \Delta m_d>\gamma\Delta m_c$.
		
		\item We have already proved above that, $\tilde m_a(t)
		\equiv \Delta m_c$ when charging and $\tilde m_a(t)
		\equiv -\Delta m_d$ when discharging. It follows from
		\eqref{eq:deltaT-dynamics}  that the temperature dynamics reduce in the
		charging scenario to \vspace{-5pt}
		\begin{linenomath*}
			\begin{align}\label{eq:deltaT_c-dynamics}
			\dot{\tilde {T}}(t)=-(\alpha+\gamma \Delta m_{c})\tilde T-\beta\Delta m_{c},
			\end{align}
		\end{linenomath*}\vspace{-5pt}
		and in the discharging scenario to
		\begin{linenomath*}
			\begin{align}\label{eq:deltaT_d-dynamics}
			\dot{\tilde {T}}(t)=-(\alpha-\gamma \Delta m_{d})\tilde T+\beta\Delta m_{d}.
			\end{align}
		\end{linenomath*}
		Both of these are linear time invariant systems driven by constant
		inputs that are asymptotically stable; stability follows from
		$\alpha>\gamma \Delta m_{d}$, which was proved above and $\alpha$, $\gamma$, and $\Delta m_c$ being positive.
		It follows from elementary linear systems
		analysis~\cite{KwakernaakSivan:72} that $\tilde{T}(t)$ converges to a
		constant steady-state value 
		irrespective of the initial condition, which is, in the charging scenario:
		$\tilde T_c^{ss}=\frac{\beta \Delta m_{c}}{-(\alpha+\gamma \Delta m_{c})}$,
		and in the discharging scenario:
		$\tilde T_d^{ss}=\frac{\beta \Delta m_d}{(\alpha-\gamma \Delta m_d)}$.
		Since $\alpha$, $\beta$, $\gamma$, $\Delta m_{c}$ and $\Delta m_{d}$ are all
		positive, $\tilde T_c^{ss}<0$, and the fact that $\tilde T_d^{ss}>0$ follows
		from Proposition~\ref{prop:properties}\ref{propitem:alpha-inequality},
		proved above. \ \\
		For the second part of the statement, we need to prove that 
		\begin{linenomath*}
			\begin{align}\label{eq:ss-TcTd-ineq-desired}
			\frac{\beta \Delta m_c}{(\alpha+\gamma \Delta m_c)}<\frac{\beta \Delta m_d}{(\alpha-\gamma \Delta m_d)}.
			\end{align}
		\end{linenomath*}
		It follows from Proposition \ref{prop:properties}\ref{propitem:alpha-inequality} that $\alpha+\gamma \Delta m_c
		> \alpha-\gamma \Delta m_c>0$. Therefore, and since all relevant
		parameters are positive,
		\begin{linenomath*}
			\begin{align}
			\frac{\beta \Delta m_c}{\alpha+\gamma \Delta m_c} &< \frac{\beta \Delta m_c}{\alpha-\gamma \Delta m_c}. \label{eq:prop3-proof-1}
			\end{align}
		\end{linenomath*}
		Again from Proposition
		\ref{prop:properties}\ref{propitem:alpha-inequality}, we get
		\begin{linenomath*}
			\begin{align}
			\frac{\beta \Delta m_c}{\alpha+\gamma \Delta m_c} & < \frac{\beta \Delta m_c}{\alpha-\gamma
				\Delta m_c} < \frac{\beta \Delta m_d}{\alpha-\gamma \Delta m_d},\label{eq:prop3-proof-2}
			\end{align}
		\end{linenomath*}
		where the second inequality follows from $\Delta m_{c}<\Delta m_{d}$.
		Thus from \eqref{eq:prop3-proof-2} we get
		the desired inequality \eqref{eq:ss-TcTd-ineq-desired}, which proves
		the second statement.
	\end{sublist}
\end{proof}
\fi
\ifArxivVersion
Now we are ready to prove Lemma \ref{lem:rte-singlecycle}. \fi
\begin{proof}[Proof of Lemma~\ref{lem:rte-singlecycle}]
	\ifArxivVersion
	Since $r_{oa}=1$, recall that we established in the proof of Proposition~\ref{prop:properties}\ref{propitem:T-ss} that $\tilde{T}(t)$ is governed by two asymptotically stable linear time invariant systems, with step inputs, \eqref{eq:deltaT_c-dynamics} and \eqref{eq:deltaT_d-dynamics}, during the charging and discharging half-periods respectively. Consider first the up/down scenario, with initial condition $\tilde{T}(0) = 0$. By solving the two differential equations~\eqref{eq:deltaT_c-dynamics}-\eqref{eq:deltaT_d-dynamics}, we obtain the temperature deviation at the end of one period of the square wave:
	\else
	Since $r_{oa}=1$, $\tilde{T}(t)$ dynamics is governed by two asymptotically stable linear time invariant systems, with step inputs as shown below: 
	\begin{align}
	\dot{\tilde {T}}(t)&=-(\alpha+\gamma \Delta m_{c})\tilde T-\beta\Delta m_{c} \text{ (charging)}, \label{eq:deltaT_c-dynamics-jrnl}\\
	\dot{\tilde {T}}(t)&=-(\alpha-\gamma \Delta m_{d})\tilde T+\beta\Delta m_{d} \text{ (discharging)} \label{eq:deltaT_d-dynamics-jrnl}.
	\end{align}
	The derivation can be found in the proof of Proposition~\ref{prop:properties}\ref{propitem:T-ss} in \cite{raman:2018round}; we omit the details here. Consider first the up/down scenario, with initial condition $\tilde{T}(0) = 0$. By solving the two differential equations mentioned above, we obtain the temperature deviation at the end of one period of the square wave:
	\fi
	\begin{align*}
	\tilde{T}(2t_p) &= \frac{-\beta \Delta m_{c}e^{-(\alpha - \gamma \Delta m_d)t_p}(1-e^{-(\alpha+\gamma \Delta m_c)t_p})}{\alpha + \gamma \Delta m_c} \\
	& \quad + \frac{\beta \Delta m_d(1-e^{-(\alpha-\gamma \Delta m_d)t_p})}{\alpha - \gamma \Delta m_d}.
	\end{align*}
	By hypothesis, $(\alpha -\gamma \Delta m_d)t_p<(\alpha + \gamma \Delta m_c)t_p\ll1$, so we can use a
	first order Taylor expansion to get the following approximation:
	\begin{align}
	\label{eq:T2tp_updown_approx}
	\tilde{T}(2t_p) \approx t_p\beta \left(  \Delta m_d - \Delta m_c
	e^{-(\alpha-\gamma \Delta m_d)t_p} \right).
	\end{align}
	Since $\Delta m_c < \Delta m_d$ and $(\alpha-\gamma \Delta m_d)>0$, we get $\tilde{T}(2t_p)
	>0$. This is possibility 1 shown in
	Figure~\ref{fig:charge-then-discharge-cases-combined}: $t_d=t_p$, while $t_c
	=t_p+t_{recov}$ for some $t_{recov}>0$. It follows from
	\eqref{eq:rte-simple} that $\eta_\rte<1$. Consider second the down/up scenario, with initial condition $\tilde{T}(0)
	= 0$. The corresponding expression becomes:                                                                      
	\begin{align*}
	\tilde{T}(2t_p) &= \frac{\beta \Delta m_d e^{-(\alpha + \gamma \Delta m_c)t_p}(1-e^{-(\alpha-\gamma \Delta m_d)t_p})}{\alpha - \gamma \Delta m_d} \\
	& \quad - \frac{\beta \Delta m_c(1-e^{-(\alpha+\gamma\Delta m_c)t_p})}{\alpha + \gamma \Delta m_c}.
	\end{align*}
	A similar approximation gives:
	\begin{align}
	\label{eq:T2tp_downup_approx}
	\tilde{T}(2t_p) \approx t_p\beta \left(   e^{-(\alpha+\gamma \Delta m_c)t_p} \Delta m_d - \Delta m_c
	\right).
	\end{align}
	As long as $t_p>t_p^*$, we have $e^{-(\alpha+\gamma \Delta m_c)t_p} <  \Delta m_c/\Delta m_d$, and thus $\tilde{T}(2t_p) <0$. 
	\ifArxivVersion
	This is possibility 2 shown in Figure~\ref{fig:discharge-then-charge-cases-combined}:
	\else
	In this case we have
	\fi
	$t_c=t_p$, while $t_d=t_p+t_{recov}$ for some $t_{recov}>0$. It now follows from \eqref{eq:rte-simple} that $\eta_\rte>1$. However, if $t_p< t_p^*$, then $e^{-(\alpha+\gamma \Delta
		m_c)t_p} >  \Delta m_c/\Delta m_d$, and we have $\tilde{T}(2t_p) > 0$.
	\ifArxivVersion
	This is possibility 1 shown in Figure~\ref{fig:discharge-then-charge-cases-combined}:
	\else
	In this case we have
	\fi
	 $t_d=t_p$, while $t_c =
	t_p+t_{recov}$ for some $t_{recov}>0$, and it follows from
	\eqref{eq:rte-simple} that $\eta_\rte<1$.
\end{proof}

\begin{proof}[Proof of Lemma \ref{prop:T-bounded}]
	\ifArxivVersion
	Recall that we established in the proof of Proposition~\ref{prop:properties}\ref{propitem:T-ss} that $\tilde{T}(t)$ is governed by two asymptotically stable, linear, time invariant systems, with step inputs, \eqref{eq:deltaT_c-dynamics} and \eqref{eq:deltaT_d-dynamics}, during the charging and discharging half-periods respectively.
	\else
	As mentioned in the proof of Lemma~\ref{lem:rte-singlecycle}, $\tilde{T}(t)$ is governed by two asymptotically stable, linear, time invariant systems, with step inputs, \eqref{eq:deltaT_c-dynamics-jrnl} and \eqref{eq:deltaT_d-dynamics-jrnl}, during the charging and discharging half-periods respectively.
	\fi
	It follows from elementary linear systems theory that the step response of a stable first-order LTI system monotonically increases (or decreases, depending on the initial condition) towards the steady-state value. Therefore, in the time interval $[0, t_p]$, the maximum value of $|\tilde T(t)|$ will be (depending on whether the system is charging or discharging) 
	\begin{linenomath*}
		\begin{align}\label{eq:Tmax-02tp}
		\lvert\tilde T(t)\rvert &\leq max\Big\{ \lvert\tilde
		T(0)\rvert,|\tilde T_c^{ss}|,|\tilde
		T_d^{ss}| \Big\},\quad \forall t \in [0,t_p].
		\end{align}
	\end{linenomath*}
	\ifArxivVersion
	The value of $\tilde{T}(t_p)$ will serve as the initial condition to the LTI dynamics that govern $\tilde{T}(t)$ during the interval $[t_p,	2t_p]$, which is either~\eqref{eq:deltaT_c-dynamics} or	\eqref{eq:deltaT_d-dynamics}. 
	\else
	The value of $\tilde{T}(t_p)$ will serve as the initial condition to the LTI dynamics that govern $\tilde{T}(t)$ during the interval $[t_p,	2t_p]$, which is either~\eqref{eq:deltaT_c-dynamics-jrnl} or	\eqref{eq:deltaT_d-dynamics-jrnl}.
	\fi
	Using the same argument, we see that the maximum value of $|\tilde T(t)|$ in this time interval will satisfy
	\begin{linenomath*}
		\begin{align*}
		\lvert\tilde T(t)\rvert &\leq max\Big\{ \lvert\tilde T(t_p)\rvert,\lvert\tilde T_c^{ss}\rvert, |\tilde T_d^{ss}|\Big\}, \quad \forall t\in [t_p,2t_p], \\
		& \leq max\Big\{ \lvert\tilde
		T(0)\rvert,\lvert\tilde
		T_c^{ss}\rvert, |\tilde T_d^{ss}|
		\Big\}, \quad \forall t\in [0,t_p], 
		\end{align*}
	\end{linenomath*}
	where the second inequality follows from combining the first inequality with \eqref{eq:Tmax-02tp}. Since $\tilde T(2t_p)$ serves as the initial condition for the second period $[2t_p,4t_p]$ and so on, we can repeat this argument ad infinitum, and arrive at the conclusion that $|\tilde T(t)|$, for any $t\geq0$, is bounded by the constants $\lvert\tilde T(0)\rvert,\lvert\tilde T_c^{ss}\rvert, \text{ and } |\tilde T_d^{ss}| $. Since $\tilde T(0)=0$, $max\Big\{ \lvert\tilde T(0)\rvert,|\tilde T_c^{ss}|,|\tilde	T_d^{ss}| \Big\}=max\Big\{|\tilde T_c^{ss}|,|\tilde	T_d^{ss}| \Big\}$. Therefore, $|\tilde T(t)|$, for any $t\geq0$, is bounded by the constants $\lvert\tilde T_c^{ss}\rvert, |\tilde T_d^{ss}| $, which proves the statement.
\end{proof}

\end{document}